\documentclass[aps,twocolumn,prl,superscriptaddress,amsmath,amssymb,amsfonts,floatfix]{revtex4-2}

\usepackage{amsmath,amsthm}
\usepackage{amssymb}
\usepackage{amsfonts}

\usepackage[dvips]{graphicx}
\usepackage{epsfig}
\usepackage{tabularx}
\usepackage{color}
\usepackage[colorlinks=true,linkcolor=blue,citecolor=magenta,urlcolor=blue]{hyperref}

\usepackage[mathscr]{eucal}
\usepackage{braket}
\usepackage{physics}
\DeclareMathOperator{\Gap}{Gap}
\DeclareMathOperator{\Halting}{Halting}
\usepackage{tikz,pgfplots}
\pgfplotsset{compat=1.18}
\usetikzlibrary{arrows.meta,calc,positioning,decorations.pathmorphing,decorations.pathreplacing,patterns,3d}

\newcommand{\AdS}{\mathrm{AdS}}
\newcommand{\loc}{\mathrm{loc}}
\newcommand{\FK}{\mathrm{FK}}
\newcommand{\clock}{\mathrm{clock}}
\newcommand{\prop}{\mathrm{prop}}
\newcommand{\init}{\mathrm{init}}
\newcommand{\halt}{\mathrm{halt}}
\newcommand{\Tile}{\mathrm{tile}}
\newcommand{\ampA}{\mathrm{amp}}
\DeclareRobustCommand{\adj}{\ensuremath{\mathrm{adj}}}
\DeclareRobustCommand{\cO}{\ensuremath{\mathcal{O}}}
\DeclareRobustCommand{\SU}{\ensuremath{\mathrm{SU}}}
\newtheorem{theoremS}{Theorem}

\newtheorem{theoremA}{Theorem}
\newtheorem{lemmaA}{Lemma}
\RequirePackage[normalem]{ulem} 
\RequirePackage{color}\definecolor{RED}{rgb}{1,0,0}\definecolor{BLUE}{rgb}{0,0,1} 
\providecommand{\DIFaddbegin}{} 
\providecommand{\DIFaddend}{} 
\providecommand{\DIFdelbegin}{} 
\providecommand{\DIFdelend}{} 
\providecommand{\DIFaddbeginFL}{} 
\providecommand{\DIFaddendFL}{} 
\providecommand{\DIFdelbeginFL}{} 
\providecommand{\DIFdelendFL}{} 
\newcommand{\DIFscaledelfig}{0.5}
\RequirePackage{settobox} 
\RequirePackage{letltxmacro} 
\newsavebox{\DIFdelgraphicsbox} 
\newlength{\DIFdelgraphicswidth} 
\newlength{\DIFdelgraphicsheight} 
\LetLtxMacro{\DIFOincludegraphics}{\includegraphics} 
\newcommand{\DIFaddincludegraphics}[2][]{{\color{blue}\fbox{\DIFOincludegraphics[#1]{#2}}}} 
\newcommand{\DIFdelincludegraphics}[2][]{
\sbox{\DIFdelgraphicsbox}{\DIFOincludegraphics[#1]{#2}}
\settoboxwidth{\DIFdelgraphicswidth}{\DIFdelgraphicsbox} 
\settoboxtotalheight{\DIFdelgraphicsheight}{\DIFdelgraphicsbox} 
\scalebox{\DIFscaledelfig}{
\parbox[b]{\DIFdelgraphicswidth}{\usebox{\DIFdelgraphicsbox}\\[-\baselineskip] \rule{\DIFdelgraphicswidth}{0em}}\llap{\resizebox{\DIFdelgraphicswidth}{\DIFdelgraphicsheight}{
\setlength{\unitlength}{\DIFdelgraphicswidth}
\begin{picture}(1,1)
\thicklines\linethickness{2pt} 
{\color[rgb]{1,0,0}\put(0,0){\framebox(1,1){}}}
{\color[rgb]{1,0,0}\put(0,0){\line( 1,1){1}}}
{\color[rgb]{1,0,0}\put(0,1){\line(1,-1){1}}}
\end{picture}
}\hspace*{3pt}}} 
} 
\LetLtxMacro{\DIFOaddbegin}{\DIFaddbegin} 
\LetLtxMacro{\DIFOaddend}{\DIFaddend} 
\LetLtxMacro{\DIFOdelbegin}{\DIFdelbegin} 
\LetLtxMacro{\DIFOdelend}{\DIFdelend} 
\DeclareRobustCommand{\DIFaddbegin}{\DIFOaddbegin \let\includegraphics\DIFaddincludegraphics} 
\DeclareRobustCommand{\DIFaddend}{\DIFOaddend \let\includegraphics\DIFOincludegraphics} 
\DeclareRobustCommand{\DIFdelbegin}{\DIFOdelbegin \let\includegraphics\DIFdelincludegraphics} 
\DeclareRobustCommand{\DIFdelend}{\DIFOaddend \let\includegraphics\DIFOincludegraphics} 
\LetLtxMacro{\DIFOaddbeginFL}{\DIFaddbeginFL} 
\LetLtxMacro{\DIFOaddendFL}{\DIFaddendFL} 
\LetLtxMacro{\DIFOdelbeginFL}{\DIFdelbeginFL} 
\LetLtxMacro{\DIFOdelendFL}{\DIFdelendFL} 
\DeclareRobustCommand{\DIFaddbeginFL}{\DIFOaddbeginFL \let\includegraphics\DIFaddincludegraphics} 
\DeclareRobustCommand{\DIFaddendFL}{\DIFOaddendFL \let\includegraphics\DIFOincludegraphics} 
\DeclareRobustCommand{\DIFdelbeginFL}{\DIFOdelbeginFL \let\includegraphics\DIFdelincludegraphics} 
\DeclareRobustCommand{\DIFdelendFL}{\DIFOaddendFL \let\includegraphics\DIFOincludegraphics} 

\begin{document}

\title{Undecidability in Spacetime Geometry via the AdS/CFT Correspondence} 
\DIFdelbegin 

\DIFdelend \author{Sameer Ahmad Mir}
\DIFdelbegin 
\DIFdelend \DIFaddbegin \email{sameerphst@gmail.com}
\affiliation{Department of Physics, Jamia Millia Islamia, New Delhi, 110025, India.}
\affiliation{Department of Computer Sciences, Asian School of Business, Noida, Uttar Pradesh, 201303, India.}
\DIFaddend 

 \author{Francesco Marino}
 \DIFaddbegin \email{francesco.marino@ino.cnr.it}
	\DIFaddend \affiliation{CNR-Istituto Nazionale di Ottica, Via Sansone 1, I-50019 Sesto Fiorentino (FI), Italy.}
	\affiliation{INFN, Sezione di Firenze, Via Sansone 1, I-50019 Sesto Fiorentino (FI), Italy}

\author{Arshid Shabir}
\DIFaddbegin \email{aslone186@gmail.com}
\DIFaddend \affiliation{Canadian Quantum Research Center, 204-3002 32 Ave, Vernon, BC V1T 2L7, Canada.}

\author{Lawrence M. Krauss}
\DIFaddbegin \email{lawrence@originsproject.org}

\DIFaddend \affiliation{The Origins Project Foundation, Phoenix, AZ 85028}

\author{Mir Faizal}
\DIFaddbegin \email{mirfaizalmir@googlemail.com}
\DIFaddend \affiliation{Canadian Quantum Research Center, 204-3002 32 Ave, Vernon, BC V1T 2L7, Canada.}
\affiliation{Irving K. Barber School of Arts and Sciences, 
  University of British Columbia Okanagan, Kelowna,
British Columbia V1V 1V7, Canada.}
\affiliation{Department of Mathematical Sciences, Durham University, Upper Mountjoy, Stockton Road, Durham DH1 3LE, UK.}
\affiliation{Faculty of Sciences, Hasselt University, Agoralaan Gebouw D, Diepenbeek, 3590, Belgium.}

\DIFdelbegin 
\DIFdelend 

\begin{abstract}
Undecidability, a hallmark of G\"odel incompleteness theorems, has recently emerged in quantum many-body physics through the spectral gap problem. We demonstrate  how this logical limitation can be holographically transmitted to a class of gravitational theories via the  AdS/CFT correspondence. By embedding a translationally invariant spin Hamiltonian with undecidable gap status into a large-N gauge theory, we generate an AdS dual in which the selection of dominant bulk saddle (Poincar\'e AdS or AdS soliton) is itself undecidable. Consequently, under standard semiclassical holographic assumptions, even determining which smooth spacetime geometry emerges from quantum gravity can be beyond the limits of computability.
\end{abstract}

\maketitle

\emph{Introduction.} The AdS/CFT correspondence provides a concrete realization of the holographic principle \cite{tHooft:1993dmi,Susskind:1994vu}, relating a quantum gravitational theory (string theory or M-theory) in AdS spacetime to a CFT on its boundary \cite{Maldacena:1997re,Gubser:1998bc,Witten:1998qj}. A powerful feature of this correspondence is its strong-weak nature: when the boundary CFT is strongly coupled, the bulk gravity theory is weakly coupled and perturbatively tractable. This makes it possible to compute correlation functions \cite{Gubser:1998bc} and transport properties in strongly interacting systems.    Beyond these applications, the correspondence represents a major advance, providing a non-perturbative formulation of string theory with specific boundary conditions. In principle AdS/CFT allows one to extract all bulk physics from the boundary CFT \cite{Aharony:1999ti,DHoker:2002nbb,Nastase:2007kj,Gubser:2009md,Hartnoll:2009sz,McGreevy:2009xe,Penedones:2016voo}, provided that a complete mapping between the observables of the two theories is established \cite{Kajuri:2020vxf}.

An important problem when exploiting this correspondence is to identify which bulk saddle dominates the gravitational path integral for a given boundary CFT or, equivalently, which semiclassical spacetime emerges, since multiple bulk saddles may be consistent with the same boundary data.
For instance, both thermal AdS and AdS black holes satisfy the same boundary conditions \cite{Hawking:1982dh}, with the dominant saddle determining whether the boundary CFT is in a confined or deconfined phase \cite{Witten:1998zw}. 
Similarly, an AdS soliton can energetically dominate Poincaré AdS under certain conditions \cite{Horowitz:1998ha}, indicating a different preferred vacuum state. Saddle-point selection underlies bulk reconstruction \cite{Hamilton:2006az,Almheiri:2014lwa}, sets the semiclassical background for Hawking radiation and the black hole information problem \cite{Maldacena:2001kr}, and influences locality in emergent spacetime \cite{Kabat:2011rz}.

The selection of the dominant bulk geometry is therefore central to holography and to the ability to use an underlying quantum gravitational theory to derive semiclassical results. Here, we show that there exists a class of AdS/CFT scenarios in which this selection is undecidable.

Our starting point is the work by Cubitt-Pérez‑García-Wolf (CPW) \cite{Cubitt:2015xsa}, who demonstrated that given a quantum many-body Hamiltonian it is, in general, computationally impossible to determine whether the system it describes is gapped or gapless. The meaning of this statement is twofold. First, no algorithm can decide in finite time whether such a system is gapped or gapless, in the same sense that the halting problem is algorithmically undecidable \cite{Turing:1937qvq}. Second, for any consistent axiomatization of mathematics, there exist specific Hamiltonians for which the presence or absence of a spectral gap cannot be decided within those axioms. This is the form of undecidability established by G\"odel's incompleteness theorems \cite{Godel:1931}.

We demonstrate here that the problem of the selection of dominant bulk geometry (Poincar\'e AdS vs AdS soliton) is dual to the spectral-gap problem of the boundary quantum many-body system. We map the local, translationally invariant Hamiltonian of Ref. \cite{Cubitt:2015xsa} into a Euclidean path-integral formulation using a Trotter decomposition, and decouple its interactions by introducing an auxiliary scalar field through a Hubbard-Stratonovich transformation. The scalar field encodes the halting bit of the Turing machine through a coupling that persists in the continuum limit. We then recast this scalar field theory as an adjoint matrix model with a non-Abelian gauge structure. This reformulation allows us to properly define the large-N limit, in which the matrix theory admits a holographic dual described by gravity in asymptotically AdS spacetime with a bulk scalar mode sourced by the boundary scalar operator. 

The key observation is that the sign of the scalar effective mass term, fixed by the same undecidable halting predicate as in the boundary spectral-gap problem, determines which smooth geometry dominates: Poincar\'e AdS (gapless phase) or the AdS soliton (gapped phase).  

\emph{Preliminaries.} We briefly recall the CPW construction \cite{Cubitt:2015xsa}, which established the undecidability of the spectral-gap problem (see also Sec. I of the Supplemental Material \cite{SM}). Consider translation-invariant, spin Hamiltonians $H_L$ with nearest-neighbor, finite-range interactions on a 2D lattice of size $L \times L$, with each site carrying a local Hilbert space dimension $d$. The spectral gap is the difference between the first excited state and the ground state, $\gamma_L=E_1(L)-E_0(L)$. A uniform spectral gap, i.e. the existence of a lower bound $\Delta>0$ such that $\inf_L\gamma_L\ge\Delta$, implies exponential decay of correlations (clustering) and stability of the quantum phase \cite{Hastings:2005pr}. On the other hand, if $\liminf_{L\to\infty}\gamma_L=0$, the system is gapless and supports arbitrarily low-energy excitations \cite{Nachtergaele:2010}.

Using the Feynman-Kitaev history state \cite{Feynman:2023yvp,A_Yu_Kitaev_1997,Kitaev2002,Kempe:2004sak}, a universal Turing machine on input $u$ is mapped to a frustration free $k$-local Hamiltonian. If the computation halts after runtime $T$, the history state carries weight $1/(T+1)$ on the halting configuration. In this case an  additional output term, $H_{\mathrm{halt}}=\delta(T)\Pi_{\mathrm{halt}}$ then shifts the ground state energy by $\delta(T)/(T + 1)$, and by zero otherwise. 

The CPW construction embeds this computational layer into the above lattice and couples it to a spectral scaling mechanism so that, if halting occurs, the effect produces a finite shift in the energy density as $L\to\infty$. This defines a computable map $u\mapsto H(u)$ that assigns to each input $u$ a set of nearest-neighbor interactions, namely a translation-invariant nearest-neighbor Hamiltonian on $\mathbb Z^2$ (here $\mathbb{Z}^2$ denotes the two-dimensional integer lattice), with the one of two alternatives: if the machine halts, there exists $\Delta(u)>0$ with $\gamma_L(u)\ge\Delta(u)$ $\forall L$ (gapped phase); if it doesn't, $\liminf_{L\to\infty}\gamma_L(u)=0$ (gapless phase). Any algorithm that would allow one to decide, from the finite description of $h(u)$, which case holds, would solve the halting problem. Because no such algorithm therefore exists, the spectral-gap promise problem is Turing-undecidable. 
This obstruction persists even in one-dimensional quantum spin chains \cite{Cubitt:2015xsa}.

G\"odel incompleteness and Turing's halting problem are two faces of the same coin: a consistent, recursively axiomatizable theory cannot be complete (or it would decide halting), and a halting decider would make such a theory complete \cite{BoolosBurgessJeffrey2007,Smullyan1992}. The obstruction for quantum chains is therefore G\"odelian: the Turing-undecidability of the spectral-gap predicate means that, for certain Hamiltonians, the statement of whether they are gapped or gapless is formally independent of any consistent computably axiomatizable theory, i.e. it can be neither proved nor refuted within that theory.

\emph{Euclidean path-integral representation.} We denote by $\vartheta(u)$ the CPW halting predicate, equal to $1$ if the machine halts on input $u$ and $0$ otherwise. As we outline here, in the continuum, $\vartheta(u)$ survives as a relevant scalar-mass operator, making the infrared (IR) choice between gapped and gapless phases algorithmically undecidable. We refer the reader to Supplemental Material \cite{SM} Sec. II for details. 

Starting from the partition function
\(
Z_{L}(\beta,u)\;=\;\Tr_{\mathcal H_{L}}\!\bigl[e^{-\beta H_{L}(u)}\bigr],
\)
a {Suzuki-Trotter decomposition \mbox{
\cite{Trotter1959,Suzuki1976,PhysRevA.28.3575} {can map the system } }\hskip0pt
} to a $(2{+}1)$D classical statistical model, accurate up to $\mathcal O(a_{\tau}^{2})$ with $a_{\tau}=\beta/N_\tau$. Here, $N_\tau$ denotes the number of discrete time slices in the Euclidean-time direction, so that the continuum limit corresponds to $N_\tau \to \infty$. 
The Lieb-Robinson bound \cite{Lieb:1972wy} ensures a finite signal velocity (compatible with relativity) but does not by itself fix the dynamical exponent $z$ that sets the anisotropic scaling between space and time. We restrict to microscopic models that flow to a Lorentz-invariant ($z=1$) fixed point (e.g., the O($N$) Wilson-Fisher class), so that our HS continuum description applies.

A Hubbard-Stratonovich (HS) decoupling yields, after integrating out spins, an $O(N_v)$ scalar theory (with $N_v$ the number of spin components) with bare Euclidean Lagrangian 
\begin{multline}
\mathcal L_E \;=\; \frac{1}{2}(\partial_\mu\varphi^A)(\partial^\mu\varphi^A)
+ \frac{m^2}{2}\,\varphi^A\varphi^A
+ \frac{\lambda}{4!}\,(\varphi^A\varphi^A)^2\\
\;+\; g_{\rm halt}\,\varphi^A\varphi^A,
\label{eq:new_L}
\end{multline}
where $g_{\rm halt}\equiv \vartheta(u)\,\mu_h^2 $ includes the halting bit, $\mu_h$ is a scale factor with the dimensions of a mass and $A=1,\dots,N_v$.
A {one-loop renormalization group }{{analysis }}  shows that $g_{\rm halt}$ is a relevant deformation (its coupling grows under coarse-graining). The effective mass of the scalar field is $
m_{\rm eff}^2(u)=m^2+g_{\rm halt}+\Sigma_{\rm loop}(\lambda)$.
The $\Sigma_{\rm loop}(\lambda)$ encodes loop corrections to the scalar mass. Renormalizing the system by enforcing $m^2+\Sigma_{\rm loop}(\lambda)=0$, gives $m_{\rm eff}^2(u)=\vartheta(u)\,\mu_h^2$ to leading order. In this case, the effective mass is controlled entirely by the halting deformation: if $\vartheta(u)=0$ then $m_{\rm eff}^2=0$, placing the IR theory at the Wilson-Fisher critical point (gapless, with long-range/power-law correlations)~\cite{Wilson:1973jj,Bagnuls:2000ae}. If $\vartheta(u)=1$ then $m_{\rm eff}^2>0$, yielding a uniformly gapped phase with correlation length $\xi=\Delta^{-1}$ and exponential clustering. Since $m_{\rm eff}^2(u)$ depends on the uncomputable predicate $\vartheta(u)$, the distinction between gapless and gapped outcomes is algorithmically undecidable.

\emph{Matrix theory.} We promote the $O(N_v)$ vector theory to an \SU$(N_c)$ matrix theory by choosing {$N_v = N_c^2 -1$ } and repackaging the vectors as an adjoint matrix $M(x)=\sum_{A=1}^{N_v}\phi_A(x)T^A$, where {$T^A \in {\mathfrak{su}(N_c)}$ } (with $\Tr[T^A T^B]=\tfrac{1}{2}\delta^{AB}$){. Gauging }  this global symmetry {to $SU(N_c)$ } introduces a non-Abelian gauge field $A_\mu$ with covariant derivative {$D_\mu M=\partial_\mu M -i[A_\mu,M]$ } and field strength {$F_{\mu\nu}=\partial_\mu A_\nu-\partial_\nu A_\mu -i[A_\mu,A_\nu]$ } (see Supplemental Material \cite{SM}, Sec. III for details). 
The adjoint action is
\begin{multline}
  S_{\mathrm{adj}}
  = N_{c}\!\int d^{3}x\,\Tr\!\Bigl[(D_{\mu}M)^{2}+m^{2}M^{2}+\lambda M^{4}
      +\tfrac{1}{4g^{2}}F_{\mu\nu}^{2}\Bigr]\\
  + N_{c}\!\int d^{3}x\,g_{\mathrm{halt}}\Tr(M^{2}),
  \label{eq:Sadj}
\end{multline}
with fixed ’t Hooft coupling $\lambda =g^{2}  N_{c}$ \cite{tHooft:1973alw}. 

This choice ensures that the planar (large-$N_c$) limit of the matrix model reproduces the large-$N_v$ behavior of the HS vector theory introduced above.
At large $N_c$, vacuum diagrams are organized by genus $h$, with contributions scaling as $N_c^{2-2h}$. As a consequence, connected correlators factorize at leading order
 $
  \bigl\langle \Tr M^{k}\,\Tr M^{\ell}\bigr\rangle
  = \bigl\langle \Tr M^{k}\bigr\rangle\bigl\langle \Tr M^{\ell}\bigr\rangle
    + \mathcal O(N_{c}^{-2}),
  \label{eq:factorization}
 $
so that a single planar saddle controls the dynamics \cite{Witten1979}. Standard $d{=}3$ renormalization gives a finite planar one-loop shift $\Sigma_{\rm loop}(\lambda ,\Lambda)=\mathcal O(1)$, ({where } $\Lambda$ {is the } UV cutoff) giving the classical quadratic action
\begin{equation}
  S_{\mathrm{cl}}
  = N_{c}\!\int d^{3}x\,\Tr\!\bigl[(\partial_{\mu}M)^{2}+m_{\mathrm{eff}}^{2}M^{2}\bigr] \, .
  \label{eq:meff}
\end{equation}
Here, $\Sigma_{\rm loop}$ is the finite planar one-loop self-energy, and imposing $m^2+\Sigma_{\rm loop}=0$ gives $m_{\mathrm{eff}}^{2}=g_{\mathrm{halt}}$ to leading order.

The self-energy has a genus expansion
\begin{equation}
  \Pi(k)=\sum_{h=0}^{\infty} N_{c}^{-2h}\,\Pi_{h}(k,\lambda ,\Lambda),
  \label{eq:genusPi}
\end{equation}
with planar part $\Pi_0=\mathcal O(1)$ and $\Pi_{h}=\mathcal O(N_{c}^{-2h})$ for $h\ge 1$ \cite{Coleman1985}. Therefore, non-planar terms cannot significantly affect $m_{\mathrm{eff}}^{2}$, and in the planar limit the IR phase is still determined by the undecidable predicate $\vartheta(u)$.

Equivalently, examining the problem in terms of the spectral gap, for a finite cubic box of side $L$ with periodic boundary conditions, the propagator and dispersion relation are, respectively (for details see Supplemental Material \cite{SM})
\begin{equation}
  G(p)=\frac{1}{p^{2}+m_{\mathrm{eff}}^{2}},\qquad
  E(p)=\sqrt{p^{2}+m_{\mathrm{eff}}^{2}} \, .
  \label{eq:PRL-propagator}
\end{equation}
The spectral gap for nonzero momentum modes is $\gamma_L^{(p\neq0)}(u)=\sqrt{(2\pi/L)^2+m_{\rm eff}^2(u)}$. In the thermodynamic limit, $\Delta(u)=\lim_{L\to\infty}\gamma_L(u)=m_{\rm eff}(u)$.
A positive $\Delta(u)$ implies exponentially decaying correlations with correlation length $\xi=\Delta^{-1}$, and a stable gapped phase. 
When $m_{\mathrm{eff}}^{2}(u)=0$, correlations display a power-law behavior corresponding to an IR CFT with dynamical exponent $z{=}1$. 

The halting-term is again crucial in determining which regime is realized: $\vartheta(u)=0$ yields criticality; if $\vartheta(u)=1$, then $m_{\mathrm{eff}}^{2}(u)\ge c>0$ uniformly in $L$, leading to a gapped phase (see Supplemental Material \cite{SM}). Since no Turing machine can infer the behavior of $\vartheta(u)$ from microscopic data \cite{Cubitt:2015xsa}, undecidability holds even for the large-$N_{c}$ continuum matrix theory. As we will see, the question of which phase ultimately arises corresponds holographically to the selection of the dominant bulk geometry. 

\emph{Ultraviolet completion.} 
Having shown that the halting deformation survives at large \(N_c\) and governs the IR behavior, we now justify its microscopic origin. 
By performing a real-space coarse-graining and introducing a HS field to linearize the inter-block couplings, one obtains an adjoint $SU(N_c)$ gauge-scalar completion with an emergent non-Abelian connection and Yang-Mills kinetic term \(\mathrm{Tr}\,F_{\mu\nu}F^{\mu\nu}\).
As shown in (Sec.~IV) of \cite{SM}, the Yang–Mills sector arises naturally under this construction, while the halting deformation contributes only through the scalar mass shift \(m^2\!\to m_{\mathrm{eff}}^2\). To leading order the halting deformation enters only through the mass shift and leaves the gauge kinetic term unchanged. Hence, the only relevant quantities are \(m_{\mathrm{eff}}^2\) and the large-\(N_c\) scaling, where \(m_{\mathrm{eff}}^2\) is determined by the undecidable predicate \(\vartheta(u)\).

 Confinement is probed by the Wilson loop \(W(C)=\frac{1}{N_c}\,\mathrm{Tr}\,\mathcal{P}\exp\!\bigl(i\oint_C A_\mu dx^\mu\bigr)\), whose area law encodes the string tension. At fixed 't~Hooft coupling \(\lambda  =g^2  N_c\), connected vacuum diagrams of genus \(h\) scale as \(N_c^{2-2h}\); hence nonplanar terms are \(O(N_c^{-2})\) and do not affect the leading behavior of \(m_{\mathrm{eff}}^2\) (see ribbon-graph analysis in Sec.~III of \cite{SM}). The partition function admits the large-$N_c$ expansion \(\ln Z(N_c,\lambda  )=\sum_{h\ge 0} N_c^{2-2h} F_h(\lambda  )\) which shows that the planar sector \(h=0\) dominates.
 Interpreted holographically, this is a string-loop expansion with coupling \(g_s\sim 1/N_c\), while \(\lambda  \) sets the effective string tension. Hence, the planar limit corresponds to classical bulk gravity in AdS/CFT and \(1/N_c\) terms encode string loop corrections \cite{tHooft:1973alw,Witten1979}. Crucially, any \(N_c\)-independent deformation, in particular the halting term \(g_{\rm halt}\) in \(m_{\mathrm{eff}}^2\), survives at planar order and determine the classical bulk geometry, while \(1/N_c\) effects are parametrically subleading and cannot overturn this shift.

\emph{Holographic duality and competing geometries.} In the holographic planar limit, the stress tensor of adjoint matter scales as $\langle T_{\mu\nu}\rangle \sim O(N_c^2)$,
while the AdS/CFT matching fixes the gravitational coupling through $L^{d-1}/G_{d+1}\sim N_c^{2}$. Their product  $G_4\langle T_{\mu\nu}\rangle$ is therefore $O(1)$, indicating that the halting-bit deformation produces a classical backreaction on the geometry.
Accordingly, our analysis is performed at leading classical order in supergravity, with higher-derivative and loop effects subleading.
 
We now examine how undecidability manifests in the competition between semiclassical bulk geometries. In the $N_c\!\to\!\infty$ limit, the source for the operator $\mathcal O=\Tr M^2$ in the boundary theory is implemented by GKPW prescription,
\begin{equation}
  Z_{\rm QFT}[J]=\exp\!\Big[-S_{\rm grav}^{\rm on\text{-}shell}\!\big(\Phi\!\mid_{z=0}=J\big)\Big],
  \label{eq:GKPW-PRL}
\end{equation}
where the bulk action includes the Einstein-Hilbert term, the Gibbons-Hawking-York boundary term and counterterms:
\begin{multline}
  S_{\rm grav}
  =\frac{1}{16\pi G_4}\!\int_{\mathcal M}\! d^4x\sqrt{-g}\Big(R+\frac{6}{L^2}\Big)\\
   +\frac{1}{8\pi G_4}\!\int_{\partial\mathcal M}\! d^3x\sqrt{-\gamma}\,K
   +S_{\rm ct}[\gamma] \, .
  \label{eq:Sgrav-PRL}
\end{multline}
Holographic renormalization ensures that the boundary stress tensor is traceless for a flat 3D boundary ($d{=}3$) \cite{Henningson:1998ey,Skenderis2002}, as expected for a CFT. Validity of the semiclassical approximation for the bulk geometry requires that the AdS radius is much larger $L\!\gg\!\ell_s$ ($\lambda_{\rm tH}\!\gg\!1$, $\ell_s^2/L^2\!\sim\!\lambda_{\rm tH}^{-1/2}$) to suppress $1/N_c^2$ loops {(where $\ell_s$ is the fundamental string length) } \cite{Polchinski:1998rq,Polchinski:1998rr,Klebanov:2000me}. 

The source $J$ couples to a bulk scalar $\Phi$ satisfying $(\Box_g-m_\Phi^2)\Phi=0$ with near-boundary expansion
 $ \Phi=z^{3-\Delta}\phi_{(0)}(x)+z^\Delta \phi_{(2\Delta-3)}(x)+\cdots,\,\ $  and $
  m_\Phi^2L^2=\Delta(\Delta-3),
  \label{eq:FG-PRL}
 $
 where $\phi_{(0)}\!\equiv\!J$ and $\phi_{(2\Delta-3)}\!\propto\!\langle\Tr M^2\rangle$ \cite{Freedman:1998tz}. Tuning $\phi_{(0)}{=}0$ gives Poincaré AdS$_4$,
\begin{equation}
  ds^2=\frac{L^2}{z^2}\big(dz^2+\eta_{\mu\nu}dx^\mu dx^\nu\big) \, ,
  \label{eq:PAdS-PRL}
\end{equation}
which represents the gapless critical phase. Turning on the halting deformation, $\vartheta(u)=1$ ({$\phi_{(0)}\!\neq\!0$)} drives $\Phi$ away from the critical point and can give rise, within the $(T,L_\theta)$ window discussed below, to the Euclidean AdS$_4$ soliton geometry \cite{Horowitz:1998ha}.
 \begin{multline}
  ds^2=\frac{r^2}{L^2}\!\big(d\tau^2+dx^2+f(r)d\theta^2\big)+\frac{L^2}{r^2 f(r)}dr^2,\\
  f(r)=1-\Big(\frac{r_0}{r}\Big)^3,\quad
  \theta\sim\theta+\frac{4\pi L^2}{3r_0}.
  \label{eq:soliton-PRL}
\end{multline}
Here, the boundary topology is $S^1_\tau(\beta)\times\mathbb R_x\times S^1_\theta(L_\theta)$, with the soliton caps $S^1_\theta$ (contractible) and the black brane caps $S^1_\tau$ (thermal). 

The corresponding free-energy densities per unit $x$-length are
$F_{\rm sol}=-c_s/L_\theta^{2}$ and $F_{\rm bb}=-c_T\,T^{3}L_\theta$, which coincide at the critical temperature
 \(
T_c(L_\theta)=\Big(\tfrac{c_s}{c_T}\Big)^{1/3}\,\frac{1}{L_\theta}.
 \)
A relevant boundary (multi-trace) deformation proportional to \(\vartheta(u)\) modifies the UV coefficients \((c_s,c_T)\) and the mixed (Robin) boundary condition of the dual scalar, shifting the critical temperature \(T_c\) while preserving large-\(N_c\) scalings and the high-\(T\) dominance of the AdS\(_4\) black brane (see \cite{SM}, Secs. V-VI). We focus in a low-temperature window with \(\beta/L_\theta \gg 1\), where the only two smooth, static, translation-invariant saddles are thermal Poincaré AdS\(_4\) and the AdS\(_4\) soliton. Their renormalized on-shell actions satisfy \(I[g_P]=0\) and \(I[g_{\rm sol}]<0\). In the undeformed theory the transition occurs at \(T_c \sim 1/L_\theta\), within this window \(\vartheta(u)=0\) selects the deconfined (Poincar\'e) branch (\(T>T_c\)) and \(\vartheta(u)=1\) selects the confining (soliton) branch (\(T<T_c\)). 
In the planar limit, large-\(N_c\) factorization implies that the deformation affects only one-point  data, e.g. sources and expectation values, without affecting the saddle competition.
 
We implement the deformation through mixed (Robin) boundary conditions for the bulk scalar, following standard GKPW holography \cite{Witten:2001ua,Skenderis2002}. Near the boundary $z\to0$, 
$
\Phi(z,x)=z^{d-\Delta}\phi_{(0)}(x)+z^{\Delta}\,\frac{\langle\mathcal O(x)\rangle}{2\Delta-d}+\cdots$ and $
m_{\rm bulk}^{2}L^{2}=\Delta(\Delta-d),
\label{eq:FG-BF-final}
 $
so $\phi_{(0)}$ sources $J$ and $\Delta$ obeying the Breitenlohner-Freedman relation \cite{Breitenlohner:1982jf}. 
The mixed boundary condition
 $
\phi_{(2\Delta-d)}(x)=\kappa_{\rm bdy}\,\phi_{(0)}(x),
 $ implements the double-trace deformation. 
Matching the stress tensor fixes $L^{d-1}/G_{d+1} \sim N^{2}_c$ \cite{Henningson:1998ey}, so that $c_{T}\propto (L^{d-1}/G_{d+1})\sim N_c^{2}$ and $\langle TT\rangle$ scales accordingly \cite{Balasubramanian:1998sn}. 

The classical supergravity regime requires both large $\lambda =g ^{2}N_c\gg1$, ensuring that higher-derivative corrections are small, and large $N_c$.   
In this regime, the halting deformation enters only through the boundary coupling
$\partial_z \Phi\big|_{z=0} \;=\; \kappa_{\rm bdy}\,\varphi_{(0)}\,, \kappa_{\rm bdy}\propto g_{\rm halt}\,L$
that eventually determines which bulk phase is realized. Since the halting predicate is uncomputable, the resulting choice of dominant bulk geometry is itself undecidable within the holographic picture (see  \cite{SM}, Sec. V).

\emph{Vacuum selection on $S^1_\beta\times\mathbb R^2$ (Casimir Test).} As an explicit example, consider the boundary metric
\(
  \gamma^{(0)}_{ij}dx^{i}dx^{j}=d\tau^{2}+dx^{2}+dy^{2}, \, 0\le\tau<\beta,
\)
which describes a spatially infinite system compactified on a Euclidean time circle of circumference $\beta$. If the theory is gapless, the vacuum is {scale-and } translation-invariant with vanishing expectation value of the stress tensor $\langle T_{ij}\rangle=0$. The corresponding bulk saddle is Poincar\'e AdS$_4$, where the Brown-York tensor also vanishes \cite{Osborn:1993cr,Henningson:1998ey}. If the spectrum is gapped, compactification on $S^1_\beta$ gives a negative Casimir energy. For a free massive scalar of mass $m_*$ with periodic boundary conditions for $\gamma^{(0)}_{ij}$, zeta-function methods \cite{Elizalde:1994gf,Elizalde:2007du} give
\begin{equation}
  \varepsilon_{\mathrm{Cas}}(\beta,m_{*})
  =-\frac{m_{*}^{2}}{2\pi^{2}\beta}\sum_{n=1}^{\infty}\frac{K_{1}(m_{*}n\beta)}{n} \, < 0
  \label{eq20}
\end{equation}
with $K_1$ the modified Bessel function. As expected, the Casimir energy vanishes as $m_*\beta\!\to\!\infty$. The dual saddle is in this case the AdS$_4$ soliton \cite{SM}, whose Brown-York tensor matches \eqref{eq20} at leading order in $G_4$
{$
  \langle T_{\tau\tau}\rangle_{\mathrm{sol}}
  =-\frac{r_{0}^{3}}{16\pi G_{4}L^{4}},    \langle T_{xx}\rangle=\langle T_{yy}\rangle=-\tfrac12\langle T_{\tau\tau}\rangle,
  \label{eq21}
$ } with $r_{0}=\frac{4\pi L^{2}}{3L_{\theta}}$, ensuring the expected tracelessness \cite{Horowitz:1998ha}. Putting these results together lead to the following cases:
\begin{equation}
  \langle T^{i}{}_{i}\rangle=
  \begin{cases}
    0, & m_{\mathrm{eff}}^{2}=0 \\
    -\dfrac{3\,r_{0}^{3}}{16\pi G_{4}L^{4}}<0, & m_{\mathrm{eff}}^{2}>0 
  \end{cases}
\end{equation}
Therefore, $\vartheta(u)$ determines whether the vacuum stress tensor vanishes (gapless CFT/ Poincar\'e AdS$_4$) or carries a negative Casimir energy (gapped CFT/AdS$_4$ soliton). This stress-tensor makes the holographic undecidability concrete: the boundary signature of the vacuum, just like the bulk geometry itself, is fixed by an uncomputable predicate (see Sec.~V of  the Supplemental Material \cite{SM}). Under standard holographic assumptions (large $N_c$, large $\lambda $, classical gravity, and static translation-invariant bulk geometries), deciding the dominant Euclidean AdS$_4$ saddle on $S^1_\beta\times\mathbb R^2$ is therefore algorithmically undecidable (see \cite{SM} Sec. VI). {So, here Tarski’s external truth predicate \cite{Tarski:1936, FaizalJHAP25} would be needed to fix which semiclassical saddle (Poincar\'e AdS$_4$ or the AdS$_4$ soliton) emerges as the dominant spacetime.}

\emph{String theoretic interpretation.} A natural concern is that semiclassical holography is highly nongeneric. Large $N$ and large ’t~Hooft coupling are necessary but not, by themselves, sufficient to ensure an Einstein gravity regime. In known examples, semiclassical bulk locality is tied to additional CFT structure, in particular a sparse set of low-dimension single-trace operators together with a parametrically large gap to higher-spin and stringy excitations \cite{Heemskerk:2009pn,Heemskerk:2010ty,Afkhami-Jeddi:2016ntf}. We therefore phrase our conclusions in the explicitly conditional form adopted throughout the paper. We construct a computable family of large $N$ theories in which an undecidable microscopic input selects a relevant deformation, and we analyze the resulting competition of smooth Euclidean bulk fillings under the standard GKPW hypothesis that the theories under consideration admit a semiclassical holographic description \cite{Maldacena:1997re,Gubser:1998bc,Witten:1998qj}.
 
 The simplified model is chosen only to keep the reduction and the saddle comparison completely explicit. The underlying computational logic is not model dependent. In any holographic CFT, one follows the same sequence of steps: specify the deformation, identify the dual bulk field and the corresponding boundary data, construct the candidate saddles compatible with the prescribed boundary geometry and sources, and determine dominance by comparing the renormalized on-shell actions using standard holographic renormalization \cite{Henningson:1998gx,Skenderis2002}. In this precise sense, our construction provides a prototype for analogous computations in any holographic setting where the dual is established and the semiclassical gravitational path integral is under quantitative control, including benchmark dual pairs such as planar $\mathcal N=4$ super-Yang-Mills theory and ABJM theory \cite{Maldacena:1997re,ABJM:2008}.

From a string-theoretic viewpoint, the extra ingredient required to export the mechanism to such top-down settings is not a new dynamical sector, but a choice of boundary coupling. In the AdS/CFT dictionary, switching on a relevant deformation corresponds to prescribing boundary data for the dual bulk scalar, usually by fixing the non-normalizable mode in standard quantization and, more generally, by imposing mixed boundary conditions when multi-trace deformations are included \cite{Witten:2001ua,Hartman:2006dy}. Equivalently, one chooses a point in the space of couplings of the brane worldvolume theory. In the bulk this is implemented by turning on the corresponding supergravity mode with specified asymptotics, often interpretable in terms of moduli in the underlying brane construction. The role of the undecidable predicate $\vartheta(u)$ is then minimal, as it only selects which value of a chosen relevant coupling is realized.

To make the export to top-down settings concrete while staying within the same semiclassical logic, one may place the boundary theory on $S^{d-1}\times S^1$ (or on a spatial circle), and work in regimes where more than one smooth Euclidean bulk filling is admissible for the same prescribed boundary geometry and sources. A protected relevant deformation then shifts the renormalized on-shell action difference between the competing saddles. In planar $\mathcal N=4$ super-Yang-Mills theory a canonical example is the $\mathcal N=2^{\ast}$ mass deformation \cite{Pestun:2007rz,Pilch:2000ue,ChenLin:2016jhs}, while in ABJM one can consider supersymmetry-preserving mass deformations with a well-developed large-$N$ and holographic description \cite{Gomis:2008vc,Nosaka:2016kpq,Jang:2017rqn}. In suitable deformation families of this type, varying the relevant coupling can, in favorable cases, lead to distinct semiclassical saddles with qualitatively different infrared behavior, in direct analogy with the branch structure exhibited by our explicit computable family. Once the boundary choice is made, the remainder of the analysis reduces to the standard semiclassical comparison of admissible Euclidean fillings via their renormalized on-shell actions.

\emph{Conclusions.} We have shown that undecidability in the spectral-gap problem for quantum many-body systems can be faithfully transmitted, via the AdS/CFT correspondence, into the bulk geometry selection problem of semiclassical AdS gravity. To our knowledge, this provides the first concrete realization of a G\"odel-Turing type obstruction in a holographic setting. The validity of this conclusion is conditioned on standard assumptions in holography: suppression of higher-derivative and string corrections, uniqueness of smooth saddles under fixed boundary conditions, and the validity of the large-N planar limit. 

 While our result relies on a specific holographic construction, it suggests that undecidability can be a robust property of holographic duality and can constrain the very emergence of spacetime itself. A natural question is whether similar undecidable structures might arise in other holographic settings, such as black hole microstates, string compactifications, or cosmological duals. If so, undecidability could represent a pervasive feature of quantum gravity. This would imply that many aspects of spacetime structure could be beyond algorithmic reach \cite{Faizal:2025gip}, and would require a form of non-algorithmic understanding that transcends the G\"odel-Turing barrier \cite{Faizal:2024rod, FaizalJHAP25}. Just like algorithmic computation, such non-algorithmic understanding is intrinsic to nature. 

\onecolumngrid
\section{Supplemental Material}
\appendix

\section{Spectral-Gap Undecidability (Full Details)}
\label{app:spectral-gap}


This section presents in detail the CPW construction \cite{Cubitt:2015xsa}, sketched in the Sec.  {Preliminaries} of the main text.  

For each positive integer \(L\) consider the square torus
\begin{equation}
  \Lambda_L=(\mathbb Z/L\mathbb Z)^2,\qquad |\Lambda_L|=L^2.
  \label{eqA:lattice}
\end{equation}
At each site \(x\in\Lambda_L\) lives a \(d\)-level register with local Hilbert space \(\mathcal H_{\loc}\cong\mathbb C^d\). The composite space is
\begin{equation}
  \mathcal H_L=\bigotimes_{x\in\Lambda_L}\mathcal H_{\loc}.
  \label{eqA:HL}
\end{equation}
A translation-invariant nearest-neighbour model is specified by a single self-adjoint plaquette operator \(H\in\mathcal B(\mathcal H_{\loc}^{\otimes 4})\) tiled periodically
\begin{equation}
  H_L=\sum_{p\subset\Lambda_L}H_p,\qquad H_p\cong H,\qquad \|H_L\|\le \|H\|\,L^2.
  \label{eqA:HLsum}
\end{equation}

Lieb-Robinson (LR) bounds hold for such finite-range interactions: there exist \(C,v,\xi>0\) such that for disjointly supported local observables \(A,B\)
\begin{equation}
  \|[A(t),B]\|\le C\,\|A\|\,\|B\|\,|\mathrm{supp}\,A|\,
   \exp\!\Big(-\tfrac{\mathrm{dist}(\mathrm{supp}A,\mathrm{supp}B)-v|t|}{\xi}\Big),
  \label{eqA:LR}
\end{equation}
where \(A(t)=e^{\mathrm i H_L t}A\,e^{-\mathrm i H_L t}\) and \(\xi\) is a locality length \cite{Lieb:1972wy}. This bound establishes quasi-locality in the thermodynamic limit, defining a \(C^*\) algebra of observables and an emergent light cone with velocity \(v\).

Let \(E_0(L)\le E_1(L)\le E_2(L)\le\cdots\) be the eigenvalues of \(H_L\). The finite-volume spectral gap is
\begin{equation}
  \gamma_L=E_1(L)-E_0(L).
  \label{eqA:finite_size_gap}
\end{equation}
A uniform spectral gap means that there exists $\Delta>0$ with \(\gamma_L\ge\Delta\) for all \(L\). Uniform gaps imply exponential clustering (correlations \(\sim e^{-\mathrm{dist}/\ell}\), with \(\ell\le v/\Delta\)) and phase stability under sufficiently small local perturbations \cite{Hastings:2005pr,Bravyi:2011bey}. If instead
\begin{equation}
  \liminf_{L\to\infty}\gamma_L=0,
  \label{eqA:gapless_def}
\end{equation}
the family is gapless: arbitrarily low-energy, long-wavelength excitations exist along some unbounded size sequence and correlations often decay algebraically \cite{Nachtergaele:2010}.

The finite-volume gap feeds into the Gelfand-Naimark-Segal (GNS) representation in the thermodynamic limit. If \(\gamma_L\to\gamma_\infty>0\), the GNS Hamiltonian has a true gap above its cyclic vector. If \(\gamma_L\to0\), arbitrarily low-energy delocalized excitations arise at length-scale \(L\), typical of criticality or topological order.

We can now phrase the  {spectral-gap promise problem}. A computable map takes a finite binary string \(u\) to a plaquette interaction \(H(u)\), hence to \(H_L(u)\) and \(\gamma_L(u)\). The decision task is:
\begin{equation}
  \inf_{L\in\mathbb N}\gamma_L(u)>0\ \ \text{(declare  {gapped})},\qquad
  \liminf_{L\to\infty}\gamma_L(u)=0\ \ \text{(declare  {gapless})}.
  \label{eqA:promise}
\end{equation}

Let \(U\) be a universal,  {reversible} Turing machine. Given input \(u\in\{0,1\}^*\), its configuration after \(t\in\{0,\dots,T\}\) steps is \(\ket{\psi(t)}\) in a finite-dimensional register \(\mathcal H_{\mathrm{UTM}}\). Introduce a unary clock \(\mathcal H_{\clock}=\mathrm{span}\{\ket{t}\}_{t=0}^{T}\). The history state
\begin{equation}
  \ket{\Psi_{\mathrm{hist}}}=\frac{1}{\sqrt{T+1}}\sum_{t=0}^{T}\ket{t}\otimes\ket{\psi(t)}
  \label{eqA:H_state}
\end{equation}
coherently records the whole trajectory \cite{Feynman:1981tf,Feynman:1986vej,A_Yu_Kitaev_1997,Kitaev2002}.

A frustration-free, \(k\)-local Hamiltonian with \(\ket{\Psi_{\mathrm{hist}}}\) as a ground state is
\begin{equation}
  H_{\FK}=H_{\clock}+H_{\prop}+H_{\init}\ge 0,
  \label{eqA:H_FK}
\end{equation}
with a clock constraint \(H_{\clock}=\sum_{t=0}^{T-1}\Pi_{\mathrm{illegal}}^{(t,t+1)}\) penalizing violations of the unary pattern, and a propagation Laplacian
\begin{equation}
  H_{\prop}=\sum_{t=0}^{T-1}\Big(\ket{t}\!\bra{t}+\ket{t+1}\!\bra{t+1}
  -U_{t+1,t}\ket{t+1}\!\bra{t}-U_{t,t+1}^\dagger\ket{t}\!\bra{t+1}\Big),
  \label{eqA:H_prop}
\end{equation}
which couples the clock step \(t\to t{+}1\) to the corresponding reversible gate \(U_{t+1,t}\) on the data and an input projector \(H_{\init}=\Pi_{\mathrm{bad\text{-}init}}\) fixing the tape to \(u\). One verifies that \(H_{\FK}\ket{\Psi_{\mathrm{hist}}}=0\).

On the legal clock subspace, \(H_{\prop}\) reduces to the path-graph Laplacian on \(T{+}1\) nodes, with eigenpairs
\begin{equation}
  \ket{\phi_k}=\sqrt{\frac{2}{T+1}}\sum_{t=0}^{T}\!\sin\!\Big(\frac{\pi k (t+1)}{T+1}\Big)\ket{t},\quad
  \lambda_k=2\Big(1-\cos\frac{\pi k}{T+1}\Big).
  \label{eqA:path_spec}
\end{equation}
Hence the propagation gap and the FK gap scale as
\begin{equation}
  \Delta_{\prop}=\lambda_1=2\!\Big(1-\cos\frac{\pi}{T+1}\Big)=\Theta\!\big((T+1)^{-2}\big),\qquad
  \Delta(H_{\FK})\ge c_0\,\Delta_{\prop}=\Theta\!\big((T+1)^{-2}\big)
  \label{eqA:FK_gap}
\end{equation}
for some universal \(c_0\in(0,1)\) (projection-lemma arguments~\cite{Kempe:2004sak,Aharonov:2009jnw}).

To spectrally distinguish halting vs.\ non-halting, add an output/halting projector
\begin{equation}
  H_{\halt}=\delta_T\,\Pi_{\halt},\qquad
  \Pi_{\halt}=\ket{T}\!\bra{T}\otimes\ket{\psi(T)}\!\bra{\psi(T)}.
  \label{eqA:H_halt}
\end{equation}
Choosing the sign and scale of \(\delta_T\) to grow with \(T\) (polynomial or stretched-exponential) ensures that halting histories produce a finite  {energy density} shift once embedded across many independent FK chains in 2D.

Each term in \(H_{\FK}+H_{\halt}\) acts on a constant number of neighboring qudits. Standard perturbative gadget constructions reduce any \(k\)-local terms to a nearest-, 2-local form on a bounded-degree graph without closing the gap by more than a constant factor \cite{Kempe:2004sak}. These locality features underpin both physical plausibility and the QMA-completeness of the \(k\)-local Hamiltonian problem for \(k\ge2\) \cite{Kempe:2004sak}.


\begin{theoremA}[Cubitt-Pérez-García-Wolf \cite{Cubitt:2015xsa}]
There exists a Turing machine that, given any input \(u\), outputs a translation-invariant, finite-dimensional, nearest-neighbour interaction \(H(u)\) on \(\mathbb Z^2\) such that, for the torus-restricted Hamiltonians \(H_L(u)\) and their gaps \(\gamma_L(u)\) in \eqref{eqA:finite_size_gap}, exactly one holds:
\begin{equation}
  \text{(i) if the encoded UTM halts on \(u\), then } \exists\,\Delta(u)>0:\ \gamma_L(u)\ge\Delta(u)\ \forall L;
  \quad
  \text{(ii) otherwise } \liminf_{L\to\infty}\gamma_L(u)=0.
\end{equation}
No algorithm can decide from the finite description of \(H(u)\) which branch holds.
\end{theoremA}

 {Sketch of the (constructive) map \(u\mapsto H(u)\).} The FK history-state construction encodes reversible computation in the ground space of a 1D chain \cite{A_Yu_Kitaev_1997,Kempe:2004sak,Aharonov:2009jnw,AharonovKempeRegev2009}. On the other hand a hierarchical aperiodic tiling (Robinson-type) embeds many such chains coherently inside a 2D, translation-invariant spin system using only nearest-neighbour constraints. An energy-amplification gadget magnifies halting-induced local energy shifts so that they persist in the thermodynamic energy density. If a decision procedure for the spectral-gap promise problem existed, composing it with \(u\mapsto H(u)\) would decide halting, a contradiction. A 1D refinement also exists \cite{Cubitt:2015xsa}.

 At each site \(z=(x,y)\in\mathbb Z^{2}\) take
\begin{equation}
  \mathcal H_{\loc}=\mathcal H_{\mathrm{comp}}\otimes\mathcal H_{\mathrm{tile}},
  \qquad
  \dim\mathcal H_{\mathrm{tile}}=28\ \text{(Robinson prototiles)},\quad
  \dim\mathcal H_{\mathrm{comp}}=d_{\mathrm c}\ge 12,
  \label{eqA:loc}
\end{equation}
where \(\mathcal H_{\mathrm{tile}}\) is spanned by Robinson tiles \cite{Robinson1971} and \(\mathcal H_{\mathrm{comp}}\) contains work/clock/marker qubits (defined below). The many-body space is \(\mathcal H=\bigotimes_z \mathcal H_{\loc}\).

 Let \(\Pi_{\mathrm{Rob}}^{(z,z')}\) project onto a valid Robinson pair across nearest neighbours \(\langle z,z'\rangle\). Define
\begin{equation}
  H_{\Tile}=\sum_{\langle z,z'\rangle}\bigl(\openone-\Pi_{\mathrm{Rob}}^{(z,z')}\bigr)\ge 0.
  \label{eqA:Htile}
\end{equation}
Its ground-state manifold consists of perfect Robinson tilings \(\Omega_{\mathrm{Rob}}\), unique up to global translations \cite{Robinson1971}. Each tiling displays a deterministic hierarchy of concentric squares of side \(5^k\). Mark the horizontal midline of every maximal square by a distinguished tile \(\square\); sites carrying \(\square\) will host the computation.

 Enumerate \(\square\)-sites along a given horizontal line as \(\{w_0,\dots,w_{T-1}\}\), with \(T=2^{p(|u|)}\) (some fixed polynomial \(p\)). At each \(w_t\) place three qubits: a work qubit \(q_t\), a unary clock bit \(c_t\), and a domain-wall marker \(m_t\). The five-local FK Hamiltonian acts along each line:
\begin{equation}
  H_{\FK}^{(\mathrm{lines})}
  =\sum_{t=0}^{T-1}\!\bigl(H_{\prop}^{(t)}+H_{\clock}^{(t)}\bigr)+H_{\init}+H_{\halt}.
  \label{eqA:HFKline}
\end{equation}
Here \(H_{\prop}^{(t)}\) enforces the reversible update on \((q_t,q_{t+1})\) conditioned on \((c_t,c_{t+1})\); \(H_{\clock}^{(t)}\) penalizes illegal substrings \(\ket{10}\); \(H_{\init}\) fixes the leftmost \(|u|\) work qubits to \(u\); and \(H_{\halt}\) is as in \eqref{eqA:H_halt} (with a sign/scale chosen below). The marker \(m_t\) isolates the chain from non-\(\square\) background sites.

 After blocking \(5\times5\) super-tiles, all terms are translation-invariant and act on at most five neighboring blocks. Tiling and computation layers commute (each acts on disjoint tensor factors, with computation only where \(\square\) occurs). The frustration-free ground space at \(\delta_T=0\) is
\begin{equation}
  \ker H_{\Tile}\cap \ker H_{\FK}^{(\mathrm{lines})}
  =\mathrm{span}\bigl\{\ket{\mathrm{Rob}}\otimes\ket{\Psi_{\mathrm{hist}}^{(\mathrm{lines})}}: \ket{\mathrm{Rob}}\in\Omega_{\mathrm{Rob}}\bigr\}.
  \label{eqA:gs_space}
\end{equation}

 If \(U\) halts within \(T\), the history includes \(\ket{T}\) and \(H_{\halt}\) shifts the line’s ground energy by \(-\delta_T/(T+1)\). If \(U\) does not halt, \(H_{\halt}\) annihilates the line’s ground state and the line’s gap is \(\Theta(T^{-2})\) \cite{Aharonov:2009jnw}.


To ensure the halting/non-halting distinction survives at the level of  {energy density} in the thermodynamic limit, we augment the 2D model with a classical amplifier coupled to the hierarchical tiling.

\paragraph*{Amplifier registers and Hamiltonian.}
For each level \(k\ge 0\) in the Robinson hierarchy there is a unique square \(\Lambda_k\) of side \(L_k=5^k\) centred at \(c_k\). Place a two-state register \(s_k\in\{0,1\}\) at \(c_k\) and define the diagonal, ultralocal operator
\begin{equation}
  H_{\ampA}=\sum_{k\ge 0}L_k^{\,q}\,s_k,\qquad q>4\ \text{fixed}.
  \label{eqA:Hamp}
\end{equation}
\(H_{\ampA}\ge 0\) and commutes with all other terms.

\paragraph*{Coupling to halting projectors.}
Replace \(H_{\halt}\) along each marked line inside \(\Lambda_k\) by
\begin{equation}
  H_{\halt}'=-\delta\sum_{k\ge 0} s_k\,\Pi_{\halt}^{(k)},
  \label{eqA:Hhaltprime}
\end{equation}
where \(\Pi_{\halt}^{(k)}\) projects onto the terminal configuration of the line inside \(\Lambda_k\), and choose \(\delta>0\) small compared to \(L_k^{\,q}\) for all \(k\). The full 2D Hamiltonian is then
\begin{equation}
  H^{(2D)}=H_{\Tile}+H_{\FK}^{(\mathrm{lines})}+H_{\ampA}+H_{\halt}'.
  \label{eqA:H2D}
\end{equation}

\paragraph*{Density accounting.}
In a valid tiling sector there are \(\Theta(L^2/\ell^2)\) marked lines of length \(\ell\) at scale \(\ell=5^k\). Choose \(T(\ell)=\Theta(\ell)\) and \(\delta_T=\alpha (T+1)^{1+\kappa}\) with \(\kappa>0\) so each active line’s halting shift is \(-\delta_T/(T+1)\sim -\alpha (T+1)^\kappa\). Weighting the line by \(s_k=1\) and the amplifier \(L_k^{\,q}\) while normalizing by area \(L^2\) yields a per-scale contribution to the energy density proportional to \(-\alpha\,\ell^{-2}\); summing over the \(O(\log L)\) scales remains \(O(1)\) as \(L\to\infty\). In contrast, for non-halting inputs, \(\Pi_{\halt}^{(k)}\) annihilates the history states and minimizing sets \(s_k=0\), so \(H_{\ampA}\) contributes nothing.


\begin{lemmaA}[Halting inputs yield a uniform gap]
\label{lemA:finite}
If the encoded \(U\) halts on input \(u\) within \(T=2^{p(|u|)}\) steps, then there exists \(\Delta_*>0\) (independent of \(L\)) such that the torus-restricted Hamiltonians \(H_L(u)\) from \eqref{eqA:H2D} obey \(\inf_L \gamma_L(u)\ge \Delta_*\).
\end{lemmaA}

 {Proof.} In a valid tiling sector, each active line has its ground energy lowered by \(-\delta_T/(T+1)\), while excitations that (i) violate a tiling projector \((\openone-\Pi_{\mathrm{Rob}}^{(z,z')})\) or (ii) create a local clock/propagation defect cost at least a constant energy (the latter \(\ge c/T^2\) from \eqref{eqA:FK_gap}). The amplifier choice makes large-scale lines even more isolated (\(\delta_T/(T+1)\sim \alpha(T+1)^\kappa\gg 1\) for large \(T\)), so the bottleneck is the smallest base scale of the hierarchy, an \(O(1)\) constant that does not shrink with \(L\). Therefore \(\inf_L \gamma_L\ge \Delta_*>0\). \hfill\(\square\)


\begin{lemmaA}[Non-halting inputs are gapless]
\label{lemA:gapless}
If the encoded \(U\) does  {not} halt on input \(u\), then \(\displaystyle \liminf_{L\to\infty}\gamma_L(u)=0\).
\end{lemmaA}

 {Proof.} If \(U\) never halts, \(\Pi_{\halt}^{(k)}\) annihilates history states, minimization sets all \(s_k=0\), and the amplifier vanishes. The remaining low-energy sector is a direct sum of independent FK chains at all hierarchical scales; each chain of length \(T(\ell)=\Theta(\ell)\) has a gap \(\Theta(T(\ell)^{-2})\) \cite{Aharonov:2009jnw}. As \(L\to\infty\), the hierarchy includes unbounded \(\ell\), hence along some sequence of volumes the longest line’s excitation energy goes to zero, proving \(\liminf_{L\to\infty}\gamma_L=0\). \hfill\(\square\)


For a fixed input \(u\), let \(\gamma_L(u)\) be the first excitation energy of the torus-restricted model \(H_L(u)\) associated to the local interaction \(H(u)\) in \eqref{eqA:H2D}. Define the spectral-gap promise predicate
\begin{equation}
  \Gap(u)=\Bigl[\ \inf_{L\ge 1}\gamma_L(u)>0\ \Bigr].
  \label{eqA:CPW_gap}
\end{equation}
By Lemmas~\ref{lemA:finite}-\ref{lemA:gapless},
\begin{equation}
  \Gap(u)=\text{true}
  \quad\Longleftrightarrow\quad
  \Halting(u)=\text{true}.
  \label{eqA:gap_halting_equiv}
\end{equation}
The map \(u\mapsto H(u)\) is computable (it prints a finite template of local matrices with rational-algebraic entries) \cite{Cubitt:2015xsa}. If a Turing machine could decide \(\Gap(u)\) from this finite description, composing it with \(u\mapsto H(u)\) would decide \(\Halting(u)\) for every \(u\), contradicting Turing’s theorem \cite{Turing:1937qvq}. Therefore the spectral-gap promise problem for this class of 2D, finite-range, translation-invariant Hamiltonians is Turing-undecidable.


We impose periodic boundary conditions. The tiling layer may admit \(O(1)\) translational degeneracy of perfect tilings; gaps are defined relative to this ground-sector manifold and are unaffected by this \(L\)-independent degeneracy. In the halting branch, the uniform gap implies exponential clustering with \(\ell\le v/\Delta_*\) and stability under sufficiently small local perturbations \cite{Hastings:2005pr,Bravyi:2011bey}. In the non-halting branch, arbitrarily low-energy, long-wavelength excitations arise along unbounded size sequences \cite{Nachtergaele:2010}. All statements are compatible with LR bounds \eqref{eqA:LR} and a finite group velocity \(v\). 



\section{Euclidean Path-Integral Representation}
\label{app:Euclidean}

As described in the corresponding section of the main text, we translate the finite-temperature partition function of the Cubitt-Pérez-García-Wolf (CPW) Hamiltonian into a three-dimensional Euclidean field theory whose couplings retain the binary halting predicate \(\vartheta(u)\in\{0,1\}\). The steps presented in this section, in full detail, are: Suzuki-Trotter decomposition to an anisotropic \(3\)D classical model, exact Hubbard-Stratonovich (HS) decoupling, continuum limit and renormalization, and large-\(N_{c}\) back-reaction bound. 

Let \(H_{L}(u)\) be the translation-invariant, finite-range CPW Hamiltonian on the torus $\Lambda_L$.
At inverse temperature \(\beta\), the partition function is
\begin{equation}
  Z_{L}(\beta;u)=\Tr_{\mathcal H_{L}}\!\bigl[e^{-\beta H_{L}(u)}\bigr].
  \label{eqS:Z_def}
\end{equation}
Split \(H_{L}\) into two sums of commuting terms placed on disjoint plaquette sublattices (``even/odd'' checkerboard):
\begin{equation}
  H_{L}=H_{L}^{\mathrm{even}}+H_{L}^{\mathrm{odd}},\qquad 
  [H_{L}^{\mathrm{even}},H_{L}^{\mathrm{odd}}]\neq0\ \text{in general}.
  \label{eqS:even_odd}
\end{equation}
Let \(N_{\tau}\in\mathbb N\) and \(a_{\tau}\equiv \beta/N_{\tau}\). The second-order Suzuki-Trotter product formula \cite{Trotter1959,Suzuki1976,PhysRevA.28.3575} gives
\begin{align}
  e^{-\beta H_{L}}
  &=\Bigl(e^{-\frac{a_{\tau}}{2}H_{L}^{\mathrm{even}}}\,
           e^{-a_{\tau}H_{L}^{\mathrm{odd}}}\,
           e^{-\frac{a_{\tau}}{2}H_{L}^{\mathrm{even}}}\Bigr)^{N_{\tau}}
    +R_{2},\label{eqS:STS}\\
  \|R_{2}\| &\le C\,\beta\,a_{\tau}^{2}\,\bigl\|[H_{L}^{\mathrm{even}},[H_{L}^{\mathrm{even}},H_{L}^{\mathrm{odd}}]]+[H_{L}^{\mathrm{odd}},[H_{L}^{\mathrm{odd}},H_{L}^{\mathrm{even}}]]\bigr\|
  =\mathcal O(\beta a_{\tau}^{2}),
  \label{eqS:STS_error}
\end{align}
where the bound follows from the Baker-Campbell-Hausdorff series and submultiplicativity of the operator norm.

Insert resolutions of identity between every factor in \eqref{eqS:STS}. Choose an on-site orthonormal basis \(\{\ket{\sigma_{x}}\}_{\sigma_{x}=1}^{d}\) and define \(\ket{\sigma^{(n)}}=\bigotimes_{x}\ket{\sigma_{x}^{(n)}}\) at Euclidean time-slice \(n\in\{0,\dots,N_{\tau}-1\}\). Then
\begin{align}
  Z_{L}(\beta;u)
  &=\sum_{\{\sigma^{(0)},\ldots,\sigma^{(N_{\tau}-1)}\}}
    \prod_{n=0}^{N_{\tau}-1}
    \matrixel{\sigma^{(n)}}{e^{-\frac{a_{\tau}}{2}H_{L}^{\mathrm{even}}}}{\sigma^{(n)}}
    \matrixel{\sigma^{(n)}}{e^{-a_{\tau}H_{L}^{\mathrm{odd}}}}{\sigma^{(n+1)}}
    \matrixel{\sigma^{(n+1)}}{e^{-\frac{a_{\tau}}{2}H_{L}^{\mathrm{even}}}}{\sigma^{(n+1)}}
    +\mathcal O(\beta a_{\tau}^{2})
\notag\\
  &\equiv \sum_{\{\sigma\}} 
     \exp\!\Bigl[-a_{\tau}\sum_{n=0}^{N_{\tau}-1}
          \bigl(\mathcal H^{\mathrm{even}}(\sigma^{(n)})
               +\mathcal H^{\mathrm{odd}}(\sigma^{(n)},\sigma^{(n+1)})\bigr)\Bigr]
     +\mathcal O(\beta a_{\tau}^{2})\,.
  \label{eqS:classical_part}
\end{align}
where periodic boundary conditions \(\sigma^{(N_{\tau})}\equiv\sigma^{(0)}\) implement the trace. 
The second-order Suzuki-Trotter decomposition maps the quantum trace to an anisotropic 3D lattice, see Fig.~\ref{figS:trotter_lattice}.

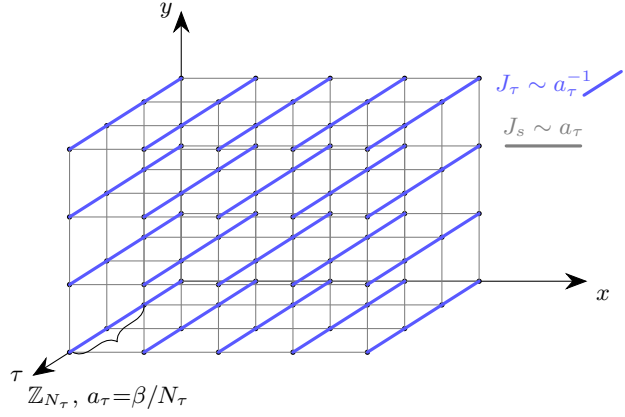
\begin{figure}[t]
\centering
\begin{tikzpicture}[scale=0.9, line cap=round, line join=round]
  \draw[-{Stealth[length=3mm]}] (0,0) -- (6,0) node[below right]{$x$};
  \draw[-{Stealth[length=3mm]}] (0,0) -- (0,4) node[left]{$y$};
  \draw[-{Stealth[length=3mm]}] (0,0) -- (-2.2,-1.4) node[left]{$\tau$};

  \def\nx{4}\def\ny{3}\def\nt{3}
  \def\dx{1.1}\def\dy{1.0}\def\dt{0.7} 

  \foreach \k in {0,...,\nt}{
    \foreach \i in {0,...,\nx}{
      \foreach \j in {0,...,\ny}{
        \coordinate (P\i\j\k) at ({\i*\dx + \k*0.0},{\j*\dy + \k*0.0});
        \coordinate (T\i\j\k) at ($ (P\i\j\k) + \k*(-0.55,-0.35)$);
        \fill (T\i\j\k) circle (1.1pt);
      }
    }
    \foreach \i in {0,...,\numexpr\nx-1\relax}{
      \pgfmathtruncatemacro{\ip}{\i+1}
      \foreach \j in {0,...,\ny}{
        \draw[gray] (T\i\j\k) -- (T\ip\j\k);
      }
    }
    \foreach \j in {0,...,\numexpr\ny-1\relax}{
      \pgfmathtruncatemacro{\jp}{\j+1}
      \foreach \i in {0,...,\nx}{
        \draw[gray] (T\i\j\k) -- (T\i\jp\k);
      }
    }
  }

  \foreach \k in {0,...,\numexpr\nt-1\relax}{
    \pgfmathtruncatemacro{\kp}{\k+1}
    \foreach \i in {0,...,\nx}{
      \foreach \j in {0,...,\ny}{
        \draw[blue!65,very thick] (T\i\j\k) -- (T\i\j\kp);
      }
    }
  }

  \draw[very thick,gray] (4.8,2.0) -- ++(1.1,0) node[midway,above]{$J_s \sim a_\tau$};
  \draw[very thick,blue!65] (6.5,3.1) -- ++(-0.55,-0.35) node[midway,left]{$J_\tau \sim a_\tau^{-1}$};

  \draw[decorate, decoration={brace, mirror, amplitude=6pt}]
        (-1.6,-1.05) -- (-0.55,-0.35)
        node[midway, below=18pt]{$\mathbb Z_{N_\tau}$, $a_\tau{=}\beta/N_\tau$};
\end{tikzpicture}
\caption{Anisotropic \(3\)D lattice \(\Lambda_L\times \mathbb Z_{N_\tau}\) from second-order Suzuki--Trotter; spatial links carry \(J_s\sim a_\tau\), temporal links \(J_\tau\sim a_\tau^{-1}\).}
\label{figS:trotter_lattice}
\end{figure}

Thus \(Z_{L}\) equals the partition function of an anisotropic \(3\)D classical model on \(\Lambda_{L}\times\mathbb Z_{N_{\tau}}\), with spatial and temporal couplings scaling as
\begin{equation}
  J_{s}\sim a_{\tau},\qquad J_{\tau}\sim a_{\tau}^{-1}.
  \label{eqS:anisotropy}
\end{equation}
The Lieb-Robinson bound \cite{Lieb:1972wy} imposes a maximal velocity \(v_{\mathrm{LR}}\) for information propagation. Taking \(a_{\tau}\to 0\) at fixed \(\beta\) rescales the temporal lattice spacing relative to the spatial one by \(v_{\mathrm{LR}}a_{\tau}\); for dynamical exponent \(z=1\), the scaling limit is isotropic up to irrelevant operators \cite{Fradkin2013,Sachdev2011}. Non-universal Trotter errors vanish as \(a_{\tau}^{2}\).

We now linearise the Boltzmann weights in \eqref{eqS:classical_part}. Consider one generic local contribution \(h(\sigma,\sigma')\) generated by \(H_{L}^{\mathrm{odd}}\). Let \(J^{A}(\sigma,\sigma')\) be bilinear densities (currents) built from on-site generators \(T^{A}\) (e.g.\ \(SU(d)\) adjoint), and \(K_{AB}\) a positive kernel. The multidimensional Gaussian identity reads
\begin{align}
  &\exp\!\Bigl[\frac{a_{\tau}}{2}\,J^{A}K_{AB}J^{B}\Bigr]
   =\int\!\frac{d^{N}\phi}{(2\pi)^{N/2}\sqrt{\det K^{-1}}}\,
     \exp\!\Bigl[-\frac{a_{\tau}}{2}\phi^{A}K^{-1}_{AB}\phi^{B}
                 +a_{\tau}\,J^{A}\phi^{A}\Bigr],\label{eqS:Gaussian}\\
  &\text{since }\ 
   \int d^{N}\phi \,e^{-\frac{1}{2}\phi^{T}A\phi + J^{T}\phi}
   =(2\pi)^{N/2}(\det A)^{-1/2} e^{\frac{1}{2}J^{T}A^{-1}J}.
\end{align}

We linearize each Boltzmann factor with an exact HS transform, Fig.~\ref{figS:HS}, introducing auxiliary fields $\phi^A$ and kernal $K$.
\begin{figure}[t]
\centering
\begin{tikzpicture}[>=Stealth, node distance=2.8cm]
  \tikzstyle{box}=[draw, rounded corners, inner sep=6pt]
  \node[box] (left) {$e^{-a_\tau\, h(\sigma,\sigma')}$};
  \node[box, right=of left] (right) {$\displaystyle \int\! d\phi^A\;
  e^{-a_\tau\left[\frac{1}{2}\phi^A K^{-1}_{AB}\phi^B - J^A(\sigma,\sigma')\phi^A\right]}$};
  --\draw[->, thick] (left) -- node[above]{HS} (right);
  \node[below=1.2cm of left] (lab1) {$J^A(\sigma,\sigma') \equiv \sigma T^A \sigma'$};
  \node[below=1.2cm of right] (lab2) {$\phi^A$ auxiliary, \(K\) positive kernel};
\end{tikzpicture}
\caption{Hubbard-Stratonovich linearisation of a generic Boltzmann factor.}
\label{figS:HS}
\end{figure}
Applying \eqref{eqS:Gaussian} plaquette-wise to all couplings in \eqref{eqS:classical_part} introduces auxiliary fields \(\phi^{A}(x,\tau)\). After the HS step, the partition sum becomes
\begin{equation}
  Z_{L}
  =\int\!\mathcal D\phi\;
   e^{-\sum_{n}\sum_{x}\frac{a_{\tau}}{2}\phi^{A}K^{-1}_{AB}\phi^{B}}
   \sum_{\{\sigma\}}
   \exp\!\Bigl[a_{\tau}\sum_{n}\sum_{x}J^{A}(\sigma,\sigma')\,\phi^{A}(x,n)\Bigr].
  \label{eqS:Z_HS}
\end{equation}
The spin sum factorises site-wise and time-slice-wise. Define the single-site generating functional
\begin{equation}
  \mathcal Z_{x}[\phi]
  =\sum_{\{\sigma_{x}^{(0)},\dots,\sigma_{x}^{(N_{\tau}-1)}\}}
   \prod_{n}
   \exp\!\Bigl[a_{\tau}\,J^{A}\!\bigl(\sigma_{x}^{(n)},\sigma_{x}^{(n+1)}\bigr)\,
                       \phi^{A}(x,n)\Bigr],
  \label{eqS:local_gen}
\end{equation}
so \(Z_{L}=\int\mathcal D\phi\,\exp\bigl[-S_{\mathrm{HS}}[\phi]\bigr]\) with
\begin{equation}
  S_{\mathrm{HS}}[\phi]
  =\sum_{n,x}\frac{a_{\tau}}{2}\phi^{A}K^{-1}_{AB}\phi^{B}
   -\sum_{x}\ln \mathcal Z_{x}[\phi].
  \label{eqS:HS_action}
\end{equation}
For slowly varying \(\phi\) one expands \(\ln\mathcal Z_{x}[\phi]\) in connected cumulants of \(J^{A}\). Up to quartic order (which is sufficient near the Gaussian fixed point),
\begin{align}
  -\ln \mathcal Z_{x}[\phi]
  &= -a_{\tau}\,\ev{J^{A}}_{0}\,\phi^{A}
     -\frac{a_{\tau}^{2}}{2}\,\ev{J^{A}J^{B}}_{0,c}\,\phi^{A}\phi^{B}\nonumber\\
  &\quad -\frac{a_{\tau}^{4}}{4!}\,\ev{J^{A}J^{B}J^{C}J^{D}}_{0,c}\,
        \phi^{A}\phi^{B}\phi^{C}\phi^{D} + \cdots,
  \label{eqS:cumulant}
\end{align}
where \(\ev{\cdots}_{0}\) denotes averages in the decoupled (\(\phi=0\)) classical model generated by the even plaquette weights. Translation invariance and symmetry force \(\ev{J^{A}}_{0}=0\). Rotational invariance in the internal space gives \(\ev{J^{A}J^{B}}_{0,c}=C_{2}\,\delta^{AB}\) and
\(\ev{J^{A}J^{B}J^{C}J^{D}}_{0,c}
 =C_{4}\,(\delta^{AB}\delta^{CD}
        +\delta^{AC}\delta^{BD}
        +\delta^{AD}\delta^{BC})\).
Passing to the continuum with \(a_{\tau}\sum_{n}\to\int d\tau\), \(a_{s}^{2}\sum_{x}\to\int d^{2}x\), gradients originate from the weak nonlocality of the HS kernel and the expansion \(\phi(x+\hat\mu)-\phi(x)\). The Euclidean Lagrangian density takes the \(O(N)\)-invariant form
\begin{equation}
  \mathcal L_{E}
  =\frac{Z_\phi}{2}(\partial_\mu \phi^A)(\partial_\mu \phi^A) \;+\; \frac{m_0^2}{2}\phi^A\phi^A
\;+\; \frac{\lambda_0}{4!}(\phi^A\phi^A)^2 \;+\; g_{\rm halt}\,\phi^A\phi^A 
  \label{eqS:Lagrangian_bare}
\end{equation}
with \(N\equiv \dim(\mathrm{Adj})\) and cutoff \(\Lambda\sim a_{\tau}^{-1}\). The last term comes from the CPW energy-amplification gadget: for halting inputs \(\vartheta(u)=1\) it provides a positive, size-dependent mass shift; for non-halting inputs \(\vartheta(u)=0\).

We canonically normalise \(\phi\to Z_{\phi}^{-1/2}\phi\) and rename renormalised couplings \(m_{0}^{2}\to m^{2}\), \(\lambda_{0}\to \lambda\).

In \(d=3\), \([\phi]=\tfrac12\), \([m^{2}]=2\), \([\lambda]=1\). Perform momentum-shell Wilsonian RG with a sharp cutoff \(\Lambda\) and step \(b=e^{\ell}>1\): split \(\phi=\phi_<+\phi_>\) with Fourier modes \(\phi_>(k)\) supported on \(\Lambda/b<|k|\le \Lambda\) and integrate out \(\phi_>\) at one loop.

\paragraph*{One-loop diagrams.} The relevant graphs are:

\(\bullet\) Tadpole for \(\phi^{2}\): 

\begin{figure}[htb]
\centering
\begin{tikzpicture}[line cap=round,line join=round]
  \draw[very thick] (-3,0) -- (-1,0);
  \draw[very thick] (1,0) -- (3,0);
  \fill (-1,0) circle (2pt);
  \fill (1,0) circle (2pt);
  \draw[very thick] (-1,0) -- (1,0);
  \draw[very thick] (0,0) circle (0.8);
  \node[below] at (-2,0) {$p$};
  \node[below] at ( 2,0) {$p$};
  \node[above] at (0,0.85) {tadpole};
\end{tikzpicture}
\caption{Tadpole diagram: one-loop correction to the two-point function (self-energy) used for \(\delta m^2\).}
\label{figS:tadpole}
\end{figure}
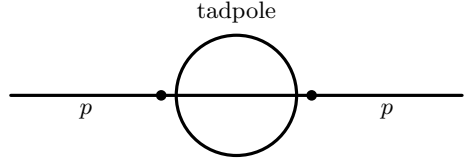

\(\bullet\) Fish for \(\phi^{4}\):
\begin{figure}[htb]
\centering
\begin{tikzpicture}[line cap=round,line join=round]
  \draw[very thick] (-3, 0.8) -- (-1, 0.4);
  \draw[very thick] (-3,-0.8) -- (-1,-0.4);
  \draw[very thick] ( 1, 0.4) -- ( 3, 0.8);
  \draw[very thick] ( 1,-0.4) -- ( 3,-0.8);
  \draw[very thick] (-1, 0.4) .. controls (0, 0.9) .. (1, 0.4);
  \draw[very thick] (-1,-0.4) .. controls (0,-0.9) .. (1,-0.4);
  \fill (-1,0.4) circle (2pt);
  \fill (-1,-0.4) circle (2pt);
  \fill ( 1,0.4) circle (2pt);
  \fill ( 1,-0.4) circle (2pt);
  \node[below] at (0,-1.2) {fish / bubble};
\end{tikzpicture}
\caption{Fish (bubble) diagram: one-loop correction to the four-point vertex used for \(\delta\lambda\).}
\label{figS:fish}
\end{figure}
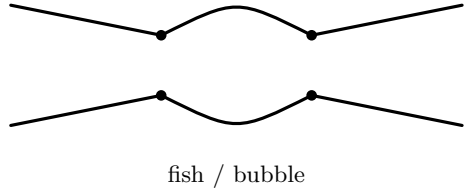

Using standard \(O(N)\) combinatorics,
\begin{align}
  \delta m^{2}
   &= \frac{N+2}{6}\,\lambda
      \int_{\Lambda/b<|k|\le \Lambda}\!\!\!\!\frac{d^{3}k}{(2\pi)^{3}}\,
      \frac{1}{k^{2}+m^{2}}
    = \frac{N+2}{6}\,\lambda\,
      \frac{1}{2\pi^{2}}
      \int_{\Lambda/b}^{\Lambda}\!dk\,\frac{k^{2}}{k^{2}+m^{2}}\nonumber\\
   &= \frac{N+2}{12\pi^{2}}\,\lambda
      \biggl[(\Lambda-\Lambda/b)
               -m\arctan\!\Bigl(\frac{\Lambda-\Lambda/b}{m}\Bigr)\biggr]\nonumber\\
   &\underset{m\ll \Lambda}{=}
      \frac{N+2}{12\pi^{2}}\,\lambda\,\Lambda\,(1-1/b)
      +\mathcal O(m),\label{eqS:delta_m2}
\end{align}
and
\begin{align}
  \delta \lambda
  &= -\frac{N+8}{6}\,\lambda^{2}
     \int_{\Lambda/b<|k|\le \Lambda}\!\!\!\!\frac{d^{3}k}{(2\pi)^{3}}\,
     \frac{1}{(k^{2}+m^{2})^{2}}
   = -\frac{N+8}{6}\,\lambda^{2}\,
     \frac{1}{2\pi^{2}}
     \int_{\Lambda/b}^{\Lambda}\!dk\,\frac{k^{2}}{(k^{2}+m^{2})^{2}}\nonumber\\
  &= -\frac{N+8}{12\pi^{2}}\,\lambda^{2}\,
     \biggl[\frac{1}{\Lambda/b+m}-\frac{1}{\Lambda+m}\biggr]
   \underset{m\ll\Lambda}{=}
     -\frac{N+8}{12\pi^{2}}\,\lambda^{2}\,\frac{1}{\Lambda}\,(1-1/b).
  \label{eqS:delta_lambda}
\end{align}

After shell integration, restore the cutoff by rescaling \(x\to x' = x/b\) so that \(k'\!=bk\). Fields scale as \(\phi'(x')=b^{(d-2)/2}\phi(x)=b^{1/2}\phi(x)\). Couplings transform as
\begin{align}
  m'^{2} &= b^{2}\bigl[m^{2}+\delta m^{2}\bigr],\qquad
  \lambda' = b\,\bigl[\lambda+\delta \lambda\bigr],\qquad
  g'_{\mathrm{halt}}= b^{2}\,g_{\mathrm{halt}},
\end{align}
where \(g_{\mathrm{halt}}\equiv \vartheta(u)\Lambda^{p(L)-2}\) multiplies \(\phi^{2}\).
Taking \(\ell=\ln b\) small and expanding to first order gives the differential RG equations
\begin{align}
  \frac{dm^{2}}{d\ell}
  &= 2m^{2} + \lim_{\ell\to 0}\frac{\delta m^{2}}{\ell}
   = 2m^{2} + \frac{N+2}{12\pi^{2}}\,\lambda\,\Lambda + \mathcal O(m\lambda),
  \label{eqS:beta_m2}\\
  \frac{d\lambda}{d\ell}
  &= \lambda + \lim_{\ell\to 0}\frac{\delta \lambda}{\ell}
   = \lambda - \frac{N+8}{12\pi^{2}}\,\lambda^{2},
  \label{eqS:beta_lambda}\\
  \frac{dg_{\mathrm{halt}}}{d\ell}
  &= 2g_{\mathrm{halt}}
     +\mathcal O(g_{\mathrm{halt}}\lambda).
  \label{eqS:beta_gh}
\end{align}
{
Define \(g \equiv \lambda/\Lambda\) and \(r \equiv m^2/\Lambda^2\).
To one loop in \(d=3\),
\begin{equation}
\frac{dg}{d\ell} = g - \frac{N+8}{12\pi^2}\,g^2 + O(g^3),\qquad
\frac{dr}{d\ell} = 2r + \frac{N+2}{12\pi^2}\,g + O(gr),
\end{equation}
and the halting deformation remains relevant:
\begin{equation}
\frac{d}{d\ell}\!\left(\frac{g_{\rm halt}}{\Lambda^2}\right)
= 2\left(\frac{g_{\rm halt}}{\Lambda^2}\right) + O(g)\,.
\end{equation}
This makes explicit that \(g_{\rm halt}\) competes with the Wilson-Fisher critical surface and drives the massive branch when \(\vartheta(u)=1\).
}
Eqs.~\eqref{eqS:beta_m2}-\eqref{eqS:beta_lambda} reproduce the canonical plus one-loop contributions for the \(O(N)\) model in \(d=3\) (the precise one-loop constants are scheme-dependent but irrelevant to the relevance classification). 
The one-loop flows are summarized schematically in Fig.~\ref{figS:RG}, for $\vartheta(u)=1$, the relevant $g_\text{halt}$ drives the flow to the massive phase.

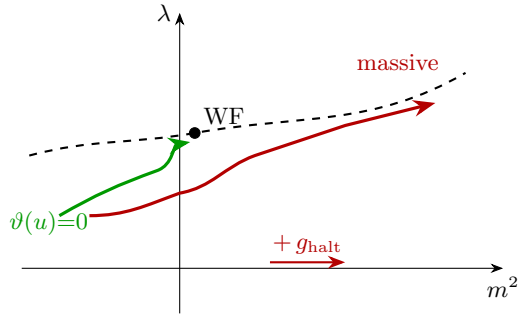
\begin{figure}[htb]
\centering
\begin{tikzpicture}[scale=1.0,>=Stealth]
  \draw[->] (-2.1,0) -- (4.3,0) node[below] {$m^2$};
  \draw[->] (0,-0.6) -- (0,3.4) node[left] {$\lambda$};

  \filldraw[black] (0.2,1.8) circle (2pt) node[above right] {WF};

  \draw[thick, dashed] (-2,1.5) to[out=15,in=190] (0.2,1.8) to[out=10,in=210] (3.8,2.6);

  \draw[->,very thick,green!60!black]
    (-1.6,0.7) to[out=30,in=200] (-0.3,1.3) to[out=20,in=190] (0.14,1.68);
  \node[green!60!black,anchor=east] at (-1.1,0.6) {$\vartheta(u){=}0$};

  \draw[->,very thick,red!70!black]
    (-1.2,0.7) to[out=0,in=200] (0.0,1.0) to[out=10,in=200] (1.0,1.5) --
    (2.2,1.9) -- (3.4,2.2);

  \node[red!70!black] at (2.9,2.75) {massive};

  \draw[->,red!70!black,thick] (1.2,0.08) -- (2.2,0.08);
  \node[red!70!black,above] at (1.7,0.08) {$+\,g_{\rm halt}$};
\end{tikzpicture}
\caption{Schematic one-loop RG flows in \((m^{2},\lambda)\) with the Wilson-Fisher fixed point. For \(\vartheta(u){=}1\), the relevant coupling \(g_{\rm halt}\) drives the flow to a massive phase. The critical-tuning trajectory (\(\vartheta{=}0\)) approaches but does not land on WF.}
\label{figS:RG}
\end{figure}

The operator \(\mathcal O=\phi^{A}\phi^{A}\) has scaling dimension \(\Delta_{\mathcal O}=1+\eta\) with \(\eta=\mathcal O(\lambda^{2})>0\) small, hence \(g_{\mathrm{halt}}\) is  {relevant} and grows as \(e^{2\ell}\) to leading order. 
If \(\vartheta(u)=0\) (non-halting input), then \(g_{\mathrm{halt}}=0\) and \(m^{2}\) can be tuned to the Wilson-Fisher critical surface, yielding a gapless continuum theory. If \(\vartheta(u)=1\) (halting input), \(g_{\mathrm{halt}}>0\) drives the flow to large positive \(m^{2}\), i.e.\ a massive phase. Thus the undecidable halting bit survives the continuum limit and reproduces CPW’s gapped/gapless dichotomy at the field-theory level. 
Placing the theory on \(\mathbb R^{2}\times S^{1}_{\beta}\),
\begin{equation}
  F_{L}(\beta,u)=-\frac{1}{\beta}\ln Z_{L}(\beta;u)
  =L^{2}\,f_{\mathrm{bulk}}\!\left[m_{\mathrm{eff}}^{2}(u)\right],\qquad
  m^2_{\rm eff}(u) \;=\; m^2 \;+\; \underbrace{g_{\rm halt}}_{=\;\vartheta(u)\,\Lambda^{\,p(L)-2}}\;+\;\Sigma_{\rm loop}(\lambda)\,.,
  \label{eqS:FreeEnergy}
\end{equation}
with \(\Sigma_{\rm{loop}}(\lambda)\sim c_{1}\lambda\Lambda+\cdots\) the one-loop mass shift from \eqref{eqS:delta_m2}. Because \(\vartheta(u)\) is Turing-undecidable, there is no algorithm that determines the sign of \(m_{\mathrm{eff}}^{2}(u)\) for a computable family \(u\mapsto H_{L}(u)\). 
Let \(\Phi\) be the bulk scalar dual to \(\mathcal O=\phi^{A}\phi^{A}\), with boundary value proportional to \(g_{\mathrm{halt}}\). Its stress tensor is
\begin{equation}
  T_{\mu\nu}(\Phi)
  =\partial_{\mu}\Phi\,\partial_{\nu}\Phi
   -\tfrac12 g_{\mu\nu}\bigl[(\partial\Phi)^{2}+m_{\Phi}^{2}\Phi^{2}\bigr].
  \label{eqS:Tmunu}
\end{equation}
Near the boundary in Fefferman-Graham gauge, \(\Phi(z,x)\sim g_{\mathrm{halt}}\,z^{\Delta_{-}}\) with \(\Delta_{-}=\frac12\bigl(3-\sqrt{9+4m_{\Phi}^{2}L^{2}}\bigr)\). {We define $\hat g \;\equiv\; g_{\rm halt}\,L^{2}$, hence near the boundary in FG gauge (\(\Delta_- = 1\))}
\begin{equation}
  \|T_{\mu\nu}(\Phi)\| \sim \hat g^{\,2}\,L^{-4}
  \label{eqS:T_scaling}
\end{equation}
The Einstein equations,
\begin{equation}
  R_{\mu\nu}-\tfrac12Rg_{\mu\nu}-\frac{3}{L^{2}}g_{\mu\nu}
  =8\pi G_{4}\,T_{\mu\nu},
\end{equation}
show that the small-backreaction condition \(\|8\pi G_4 T_{\mu\nu}\| \ll L^{-2}\) becomes, i.e.
\begin{equation}
 \hat g^{\,2} \ll \frac{L^{2}}{8\pi G_4}\ \propto\ N_c^{2}\quad
\Longleftrightarrow\quad g_{\rm halt} L^{2} \ll c\,N_c
  \label{eqS:backreaction_bound}
\end{equation} for some \(c=O(1)\) at fixed ’t Hooft coupling \cite{Maldacena:1997re,Aharony:1999ti}. Thus in the planar limit the scalar’s back-reaction is negligible and saddle selection is unaffected. Since \(g_{\mathrm{halt}}=\vartheta(u)\Lambda^{p(L)-2}=\mathcal O(1)\) (in Planck units) while \(N_{c}\to\infty\), the inequality holds parametrically. Near the boundary, the scalar profile and the large-$N_c$ back-reaction bound are illustrated in Fig.~\ref{figS:backreaction}. 

\begin{figure}[htb]
\centering
\begin{tikzpicture}[scale=1.0]
  \fill[blue!6] (0,0) rectangle (6,3.2);
  \draw[thick] (0,0) -- (6,0) node[right]{boundary $z{=}0$};
  \draw[thick] (0,3.2) -- (6,3.2) node[right]{$z\to\infty$};
  \draw[very thick, red!70!black, samples=80, domain=0.3:5.8]
    plot (\x, {2.8*exp(-0.55*\x)})
    node[right,yshift=8pt] {$\Phi(z)\sim g_{\rm halt}\, z^{\Delta_-}$};
  \node[align=left, draw, rounded corners, fill=white, anchor=west] at (1.4,2.2) {%
    $\displaystyle
    \begin{aligned}
      \|8\pi G_4\, T_{\mu\nu}(\Phi)\| &\ll L^{-2}\\
      \Rightarrow\quad g_{\rm halt}^{2} &\ll \frac{1}{L^{2}G_{4}}
        \sim \frac{N_c^{2}}{L^{2}}
    \end{aligned}$};
\end{tikzpicture}
\caption{Near-boundary decay of the bulk scalar and the large-\(N_c\) back-reaction bound \(g_{\rm halt}^2 \ll L^2/G_4 \;\sim\; N_c^2\).}
\label{figS:backreaction}
\end{figure}
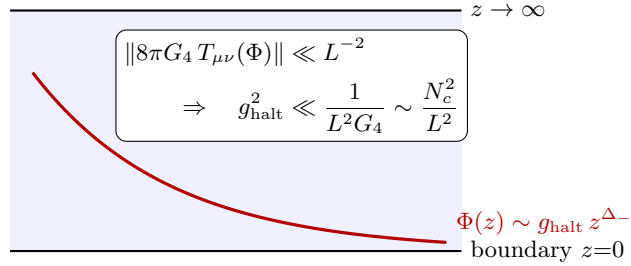

Thus in the planar limit the scalar’s back-reaction is negligible, and saddle selection between Poincaré \(\mathrm{AdS}_{4}\) and the \(\mathrm{AdS}_{4}\) soliton is unaffected.

We now analyse the adjoint-matrix sector in the large-\(N\) (planar) limit. Write \(N\equiv N_{c}\), rescale the gauge potential \(A_{\mu}\to gA_{\mu}\) so the Yang-Mills term reads \(\frac{N}{4}\Tr F_{\mu\nu}^{2}\). In double-line notation, the weight of a vacuum diagram \(G\) is \(\mathcal A_{G}\propto N^{\chi(G)}\lambda^{E}\), \(\lambda=g^{2}N\), with Euler index \(\chi=2-2h-b\) \cite{tHooft:1973alw,Witten1979}. Thus
\begin{equation}
  \ln Z(N,\lambda)=\sum_{h=0}^{\infty}N^{2-2h}F_{h}(\lambda),\qquad
  \langle\Tr M^{k}\Tr M^{\ell}\rangle
  =\langle\Tr M^{k}\rangle\langle\Tr M^{\ell}\rangle+\mathcal O(N^{-2}),
  \label{eqS:genus_expansion}
\end{equation}
so planar (\(h=0\)) diagrams dominate and factorisation holds \cite{Coleman1985}. 
’t Hooft double-line counting behind $ln\:Z=\sum_h N^{2-2h}F_h$ is depicted in Fig. S7 (planar vs non-planar).

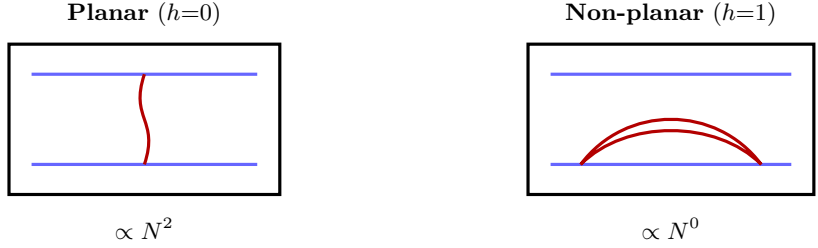
\begin{figure}[t]
\centering
\begin{tikzpicture}[scale=1.0]
  \node at (0,2.4) {\textbf{Planar} ($h{=}0$)};
  \draw[very thick] (-1.8,0) rectangle (1.8,2.0);
  \draw[very thick,blue!60] (-1.5,0.4) -- (1.5,0.4);
  \draw[very thick,blue!60] (-1.5,1.6) -- (1.5,1.6);
  \draw[very thick,red!70!black] (0,0.4) .. controls (0.2,1.0) and (-0.2,1.0) .. (0,1.6);
  \node[below] at (0,-0.2) {$\propto N^{2}$};
  \begin{scope}[xshift=7.0cm]
    \node at (0,2.4) {\textbf{Non-planar} ($h{=}1$)};
    \draw[very thick] (-1.9,0) rectangle (1.9,2.0);
    \draw[very thick,blue!60] (-1.6,0.4) -- (1.6,0.4);
    \draw[very thick,blue!60] (-1.6,1.6) -- (1.6,1.6);
    \draw[very thick,red!70!black]
      (-1.2,0.4) .. controls (-0.6,1.2) and (0.6,1.2) .. (1.2,0.4)
      .. controls (0.6,1.0) and (-0.6,1.0) .. cycle;
    \node[below] at (0,-0.2) {$\propto N^{0}$};
  \end{scope}
\end{tikzpicture}
\caption{Double-line (’t Hooft) counting: planar diagrams scale as \(N^{2}\), non-planar with one handle as \(N^{0}\).}
\label{figS:planar_nonplanar}
\end{figure}

Introduce an HS auxiliary \(\sigma\propto \Tr M^{2}\) to linearise quartic self-interactions in the adjoint scalar/matrix sector. The stationary-phase (gap) equation at leading order in \(1/N\) is
\begin{equation}
  -\partial^{2}\sigma + m^{2}\sigma +\lambda\sigma^{3}
  +\Pi_{1}^{\mathrm{planar}}\,\sigma=0,
  \label{eqS:sigma_eom}
\end{equation}
where \(\Pi_{1}^{\mathrm{planar}}\) is the planar one-loop tadpole. In \(d=3\) with dimensional regularisation and minimal subtraction,
\begin{align}
  \Pi_{1}^{\mathrm{planar}}
  &= \lambda\,(N^{2}-1)\int\!\frac{d^{3-\epsilon}k}{(2\pi)^{3-\epsilon}}\frac{1}{k^{2}+m^{2}}
   = \lambda\,(N^{2}-1)\,\frac{\mu^{\epsilon}}{(4\pi)^{\frac{3-\epsilon}{2}}}
     \Gamma\!\Bigl(\frac{\epsilon-1}{2}\Bigr)\,(m^{2})^{\frac{1-\epsilon}{2}}\nonumber\\
  &= \lambda\,(N^{2}-1)\biggl[
       \frac{c_{\Lambda}}{4\pi^{2}}\Lambda
       -\frac{m}{4\pi}
       +\mathcal O(\epsilon)\biggr],
  \label{eqS:planar_tadpole}
\end{align}
with scheme-dependent \(c_{\Lambda}\). Renormalising \(m^{2}\) absorbs the \(\Lambda\)-piece. The finite remainder defines the running mass
\begin{equation}
  m_{\mathrm{eff}}^{2}(\mu)=m^{2}+\Sigma_{1}(\lambda,\mu)+\cdots,\qquad
  \Sigma_{1}(\lambda,\mu)=\frac{\lambda\,(N^{2}-1)}{4\pi}\,(-m) + \cdots,
  \label{eqS:meff_largeN}
\end{equation}
where dots denote subleading \(1/N\) and higher-loop terms. (Any equivalent renormalisation scheme is acceptable; only the sign-stability matters below.) The effective potential in \(\overline{\mathrm{MS}}\) then reads
\begin{equation}
  \frac{V_{\mathrm{eff}}(\sigma)}{N}
  = \frac{m^{2}_{\mathrm{eff}}}{2}\sigma^{2}
    +\frac{\lambda}{4}\sigma^{4}
    +\mathcal O(N^{-2}).
  \label{eqS:Veff}
\end{equation}
For \(m^{2}_{\mathrm{eff}}>0\), the unique minimum is \(\sigma=0\) and the gap is \(\Delta=m_{\mathrm{eff}}\). For \(m^{2}_{\mathrm{eff}}=0\), the quartic stabilises the saddle, which becomes scale-invariant with anomalous dimensions governed by the Wilson-Fisher fixed point \cite{Bagnuls:2000ae}. Non-planar corrections are suppressed by \(N^{-2h}\) and cannot flip the sign of \(m_{\mathrm{eff}}^{2}\) beyond a fixed \(N\)-threshold. In AdS/CFT, the classical gravitational action acquires an overall \(N\!\propto\! L^{2}/G_{4}\) factor; \(1/N\) corrections correspond to string loops \(\sim g_{s}\sim \lambda/N\) \cite{Aharony:1999ti}. Hence the gap sign is stable in the topological expansion.

We compute the one-loop corrections to the Euclidean actions on the Poincaré \(\AdS_{4}\) saddle \(g_{\mathrm P}\) and on the \(\AdS_{4}\) soliton saddle \(g_{\text{Sol}}\). For a field \(\Phi\) of spin \(s\) and mass \(m_{\Phi}\) minimally coupled to \(g\),
\begin{equation}
  \Delta I^{(1)}_{\Phi}[g]
  =\frac{(-1)^{2s}}{2}\log\det\nolimits'\!\bigl[-\nabla_{g}^{2}+m_{\Phi}^{2}\bigr],
  \label{eqS:1loop}
\end{equation}
with zero modes removed. Using the Schwinger proper-time representation and \(\zeta\)-function regularisation,
\begin{equation}
  \log\det\nolimits'\!\bigl[-\nabla^{2}+m^{2}\bigr]
  = -\int_{0}^{\infty}\frac{ds}{s}\,\bigl[K(s)-\mathcal P\bigr]e^{-s m^{2}},
  \qquad
  K(s)=\Tr e^{-s(-\nabla^{2})},
  \label{eqS:heat}
\end{equation}
where \(\mathcal P\) projects onto zero modes. The early-time expansion is
\(K(s)=(4\pi s)^{-2}\sum_{n\ge 0} a_{n}s^{n}\); \(a_{0},a_{1}\) fix UV divergences \cite{Vassilevich:2003xt}. Since \(g_{\mathrm P}\) and \(g_{\text{Sol}}\) are asymptotically identical, counter terms renormalise them identically, only finite parts differ.

For a minimally coupled scalar on \(g_{\mathrm P}\), the spectrum of \(-\nabla^{2}\) is continuous. Following \cite{Camporesi:1994ga}, the heat kernel is
\begin{equation}
 K_P(s)
= \int \frac{d^3 p}{(2\pi)^3} \int_0^\infty dk\, \rho(k)\;
  e^{-s\left(k^2+|p|^2+\frac{9}{4L^2}\right)}
= \frac{\mathcal V_{\rm ren}(AdS_4)}{(4\pi s)^{2}}\,
  \exp\!\left(-\frac{9\,s}{4L^2}\right)
  \label{eqS:KP}
\end{equation}
with spectral density \(\rho(k)=L^{3}k^{2}/(4\pi^{2})\) and regulated boundary volume \(V_{\mathbb R^{3}}\) \footnote{
In \(D\) bulk dimensions the Schwinger kernel behaves as
\(K(s)\sim (4\pi s)^{-D/2}\sum_{n\ge 0} a_n s^n\). Our \(D=4\) case thus
gives the \((4\pi s)^{-2}\) scaling used above; the UV counterterms coincide
for the two asymptotically AdS geometries, so only finite parts differ.}
. Insert \eqref{eqS:KP} into \eqref{eqS:heat}. Subtract the \(a_{0},a_{1}\) UV pieces (holographic counterterms \cite{Skenderis2002}), the finite remainder is proportional to \((m^{2}L^{2}+2)\). For a conformally coupled scalar (\(m^{2}L^{2}=-2\)),
\begin{equation}
  \Delta I^{(1)}_{\mathrm P}=0,
  \label{eqS:DeltaIP_zero}
\end{equation}
consistent with \cite{Hartman:2006dy}. Thus scalars do not shift the Poincaré action at one loop in this case. We compare the two Euclidean saddles used throughout Poincaré AdS and the AdS$_4$ soliton in Fig.~\ref{figS:geometries}.

\begin{figure}[htb]
\centering
\begin{tikzpicture}[scale=1.0]
  \draw[thick] (0,0) rectangle (5,3.3);
  \node[above] at (2.5,3.5) {\textbf{(a) Poincar\'e AdS$_4$}};
  \draw[thick] (0,3.0) -- (5,3.0) node[right]{$z{=}0$ (boundary)};
  \foreach \x in {0.5,1.5,2.5,3.5,4.5}{
    \draw[blue!60] (\x,3.0) -- (\x,0.3);
  }
  \node at (2.5,-0.5) {$ds^{2}=\frac{L^{2}}{z^{2}}(dz^{2}+d\tau^{2}+dx^{2}+dy^{2})$};
  \begin{scope}[xshift=8.2cm]
    \draw[thick] (0,0) rectangle (5,3.3);
    \node[above] at (2.5,3.5) {\textbf{(b) AdS$_4$ soliton}};
    \draw[thick] (0.8,0.6) .. controls (1.6,0.2) and (3.4,0.2) .. (4.2,0.6);
    \draw[thick] (0.8,2.8) .. controls (2.0,3.2) and (3.0,3.2) .. (4.2,2.8);
    \draw[thick] (0.8,0.6) -- (0.8,2.8);
    \draw[thick] (4.2,0.6) -- (4.2,2.8);
    \draw[red!70!black,->] (5.3,1.7) -- (5.3,0.85) node[right] at (5.3,1.3) {$\theta$ shrinks};
    \node at (2.5,1.4) {$\tau\sim \tau+\beta,\quad r{\ge}r_{0}$};
    \node[below] at (2.5,-0.3) {smooth tip at $r{=}r_{0}$ fixes $\beta$};
  \end{scope}
\end{tikzpicture}
\caption{Two competing Euclidean saddles: (a) Poincar\'e AdS$_4$ with non-contractible $\tau$ and flat $\mathbb{R}_x\times\mathbb{R}_y$; 
(b) AdS$_4$ soliton where the  {spatial} circle $S^1_\theta$ shrinks smoothly at $r=r_0$, fixing the periodicity $\beta_\theta = 4\pi L^2/(3 r_0)$. 
The boundary topology is $S^1_\tau(\beta)\times\mathbb{R}_x\times S^1_\theta(L_\theta)$.}
\label{figS:geometries}
\end{figure}
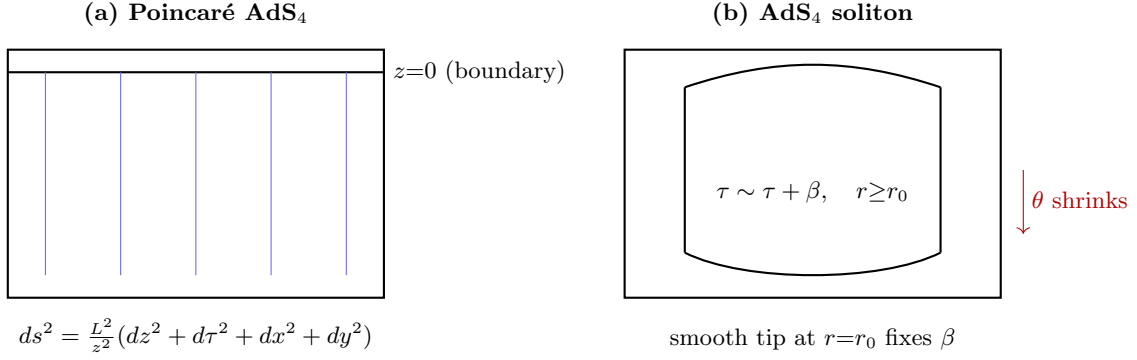


\begin{figure}[t]
\centering
\begin{tikzpicture}[scale=1.0]
  \draw[thick] (0,0) circle (2);
  \node at (0,2.5) {$S^{1}_{\beta}$};
  \foreach \ang/\lab in {0/$0$,30/$\beta$,60/$2\beta$,90/$3\beta$,120/$4\beta$,150/$5\beta$,180/$6\beta$,210/$7\beta$,240/$8\beta$,270/$9\beta$,300/$10\beta$,330/$11\beta$}{
    \fill (2*cos \ang,2*sin \ang) circle (1.2pt);
  }
  \draw[->,thick] (2,0) arc (0:60:2) node[midway, right] {images $n\beta$};
  \draw[->, thick] (3.0,0) -- (5.6,0) node[midway,above, yshift=2pt]{Poisson resummation};
  \begin{scope}[xshift=9.2cm]
    \draw[->] (-2.4,0) -- (2.8,0) node[below]{$\omega$};
    \foreach \n in {-4,-3,...,4}{
      \draw[very thick,blue!60] ({\n*0.6},0) -- ({\n*0.6},0.9);
      \node[below] at ({\n*0.6},0) {$\frac{2\pi \n}{\beta}$};
    }
    \node[above] at (0.5,1.2) {KK modes};
  \end{scope}
\end{tikzpicture}
\caption{Heat kernel on the shrinking circle via the method of images (winding along $S^1_\theta$) and its Poisson-resummed KK spectrum.}
\label{figS:images}
\end{figure}
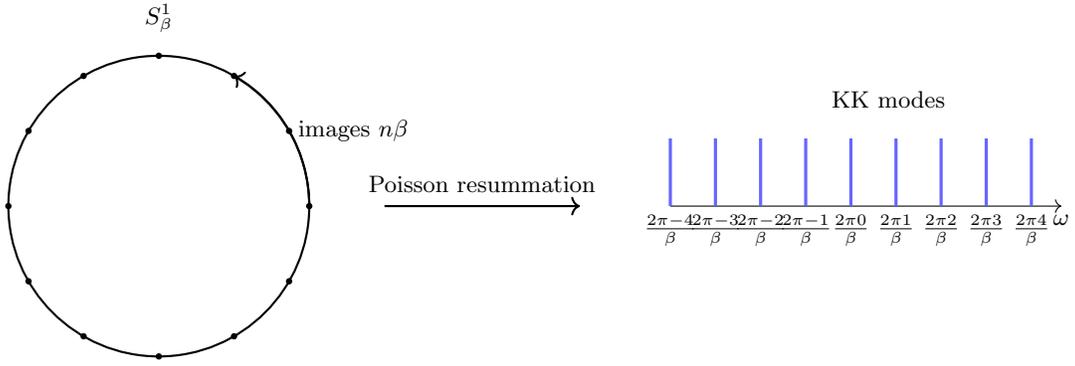

The soliton is a freely acting $\mathbb{Z}$--quotient along the spatial circle 
$\theta \sim \theta + n L_{\theta}$. The heat kernel then follows by the method of images.
\begin{equation}
K_{\rm Sol}(s) \;=\; K_P(s)\;+\; \frac{L^{3}V_x}{(4\pi s)^{2}}\,e^{-\frac{9s}{4L^2}}
\left[\frac{L_\theta}{\sqrt{4\pi s}}\sum_{n\neq 0}e^{-\frac{n^{2}L_\theta^{2}}{4s}}\right],
\label{eq:Ksol_theta_images}
\end{equation}
with \(V_{x}\) the (regulated) spatial volume along \(x\).   The method of images on $S_{\beta}^1$ and its Poisson-resummed KK spectrum are shown in Fig.~\ref{figS:backreaction}.
Poisson resumming the image sum along the shrinking circle $S^1_{\theta}$ yields
Eq.~\eqref{eq:Ksol_theta_images} and substituting into \eqref{eqS:heat} and integrating over \(s\) yields
\begin{equation}
\Delta I^{(1)}_\Phi[g_{\rm sol}]-\Delta I^{(1)}_\Phi[g_P]
\;=\;
-\,\frac{V_x\,\beta}{2\pi L_\theta}\sum_{m=1}^{\infty}\frac{m_{\rm eff}}{m}\,K_1\!\left(m\,m_{\rm eff}L_\theta\right)
\;=\;O\!\left(e^{-m_{\rm eff}L_\theta}\right),
\label{eq:new_1loop}
\end{equation}

where \(m_{\mathrm{eff}}^{2}=m^{2}+g_{\mathrm{halt}}+\Sigma_{\text{loop}}\) and \(K_{1}\) is a modified Bessel function. For \(\beta m_{\mathrm{eff}}\gg 1\), \(K_{1}(z)\sim \sqrt{\pi/(2z)}\,e^{-z}\).

Vectors decompose as \(A_{\mu}=A_{\mu}^{\perp}+\nabla_{\mu}\varphi\) with \(\nabla^{\mu}A_{\mu}^{\perp}=0\). Faddeev-Popov ghosts cancel the longitudinal sector, leaving a transverse determinant equivalent (up to multiplicity \(d_{\mathrm G}=N^{2}-1\)) to that of a scalar with an effective mass fixed by curvature \cite{Gibbons:1976ue,Christensen:1979iy}.

Metric fluctuations \(h_{\mu\nu}\) decompose into transverse-traceless tensors \(h^{\mathrm{TT}}_{\mu\nu}\), transverse vectors, and scalars. In de Donder gauge the tensor determinant reduces to a ratio of scalar-type determinants \cite{Avramidi:2000bm}. The UV divergences are governed by Seeley-DeWitt coefficients \(a_{0},a_{1}\) common to both saddles, differences arise first at \(a_{2}\), yielding a finite free-energy shift of order \(L^{3}/\beta^{3}\), suppressed against the classical \(L^{5}/(G_{4}\beta^{3})\) term when \(L^{2}\gg G_{4}\). So quantum corrections cannot reverse the classical sign of the free-energy difference between saddles when \(m_{\mathrm{eff}}^{2}>0\) and \(L^{2}\gg G_{4}\).

The finite-temperature CPW model maps to a \(3\)D Euclidean \(O(N)\) field theory with a relevant mass-like operator \(g_{\mathrm{halt}}\phi^{2}\) whose coupling is the halting bit dressed by the CPW amplifier. Under RG, \(\vartheta(u)=0\) flows to a critical (gapless) theory, while \(\vartheta(u)=1\) flows to a massive theory with \(m_{\mathrm{eff}}^{2}>0\). In the dual bulk description, back-reaction of the scalar \(\Phi\) is parametrically suppressed at large \(N_{c}\), and the free-energy difference between Poincaré \(\AdS_{4}\) and the \(\AdS_{4}\) soliton saddles is governed by the classical term with one-loop corrections subleading. Therefore the undecidable halting bit passes unaltered into the semiclassical gravitational sector, fixing the bulk saddle selection in an undecidable way, as stated in the main text.

\section{Large-\texorpdfstring{$N_c$}{Nc} Matrix Completion}
\label{app:largeN}

We now give a self-contained derivation of the large-$N_c$ matrix embedding, expanding all intermediate steps outlined in the main text.

We embed the $O(N)$ vector $\phi^A(x)$ into a traceless Hermitian matrix $M(x)\in\mathfrak{su}(N_c)$ so that gauge invariance and ’t~Hooft counting become manifest. The covariant derivative and field strength are
\begin{equation}
  D_\mu M=\partial_\mu M-i[A_\mu,M],\qquad
  F_{\mu\nu}=\partial_\mu A_\nu-\partial_\nu A_\mu-i[A_\mu,A_\nu].
\end{equation}
The gauge-invariant adjoint-matrix action is
\begin{equation}
  S_{\adj}
  =
  N_c\!\int d^3x\ \Tr\!\Big[(D_\mu M)^2+m^2 M^2+\lambda M^4+\frac{1}{4g^2}F_{\mu\nu}^2\Big]
  \ +\
  N_c\!\int d^3x\ g_{\mathrm{halt}}\Tr(M^2),
  \label{eqS:adjoint_action}
\end{equation}
with $g_{\mathrm{halt}}=\vartheta(u)\Lambda^{p(L)-2}$, $\vartheta(u)\in\{0,1\}$ is the halting bit, and $\Lambda$ is a UV cutoff of CPW provenance. The overall $N_c$ is chosen so that the ’t~Hooft coupling $\lambda =g ^2N_c$ is finite as $N_c\to\infty$.
We sketch the standard double-line counting in $d=3$ and record explicit factors. 
Write the interaction sector as a sum of single-trace vertices (e.g. $\Tr M^4$, $\Tr F^2$, $\Tr M^2 A^2$, $\Tr M [A_\mu,[A_\mu,M]]$, etc.). In the double-line (’t~Hooft) notation:
\begin{itemize}
  \item each propagator carries two colour lines and contributes a factor $\delta^{a}{}_{b}\delta^{c}{}_{d}$ which, under index contraction, yields one power of $N_c$ per closed index loop;
  \item each vertex contributes a power of the coupling; with the normalisation in \eqref{eqS:adjoint_action} a quartic $\Tr M^4$ vertex contributes $N_c\lambda$, while gauge vertices contribute $N_c(1/g ^2)$ times explicit $g $’s from rescaling $A_\mu\to g  A_\mu$ to canonical kinetic form.
\end{itemize}
For a connected vacuum ribbon graph $G$ with $V$ vertices, $E$ propagators and $F$ index loops (faces), the color factor is $N_c^{F}$ and the coupling factor is schematically $\lambda ^{E}$ after the $A_\mu\to gA_\mu$ rescaling. Using the Euler identity for a connected ribbon graph $F-E+V=\chi(G)=2-2h-b$, where $h$ is the genus and $b$ the number of boundaries (for pure vacuum $b=0$), we obtain
\begin{equation}
  \mathcal A_G \propto N_c^{F-E+V}\ \lambda ^{E}
  = N_c^{\,\chi(G)}\ \lambda ^{E}
  = N_c^{2-2h}\,\lambda ^{E}.
  \label{eqS:amplitude_scaling}
\end{equation}
Accordingly, the free energy admits the genus expansion
\begin{equation}
  \ln Z(N_c,\lambda )=\sum_{h=0}^{\infty}N_c^{2-2h}F_h(\lambda ),
  \label{eqS:genus}
\end{equation}
and $h=0$ (planar) diagrams dominate as $N_c\to\infty$. Factorisation of connected correlators follows immediately:
\begin{equation}
  \langle \Tr M^k\,\Tr M^\ell\rangle
  = \langle \Tr M^k\rangle\,\langle \Tr M^\ell\rangle + \cO(N_c^{-2}).
\end{equation}
At quadratic level in momentum space, the planar (colour-diagonal) propagator reads
\begin{equation}
  \langle M^{a}{}_{b}(p) M^{c}{}_{d}(-p)\rangle
  = \delta^{a}{}_{d}\,\delta^{c}{}_{b}\ \frac{1}{p^2+m^2 + \Pi(p)}\,,
\end{equation}
where the self-energy admits a genus expansion
\begin{equation}
  \Pi(p)=\sum_{h=0}^{\infty}N_c^{-2h}\ \Pi_h(p;\lambda ,\Lambda),\qquad
  \Pi_0=\text{planar},\ \ \Pi_{h\ge1}=\cO(N_c^{-2h}).
\end{equation}
We now compute the planar tadpole integral in $d=3$ explicitly. 
The one-loop mass correction from a quartic interaction is
\begin{align}
  \delta m^2_{\mathrm{planar}}
  &= C\,\lambda  \int_{\abs{\mathbf k}\le\Lambda}\!\frac{d^3 k}{(2\pi)^3}\ \frac{1}{k^2+m^2}\,,
  \label{eqS:tadpole0}
\end{align}
with a positive combinatorial factor $C$ (fixed by the vertex choice; its exact value is immaterial for the sign analysis). Evaluate
\begin{align}
  I_3(m,\Lambda)
  =\int_{\abs{\mathbf k}\le\Lambda}\!\frac{d^3 k}{(2\pi)^3}\ \frac{1}{k^2+m^2}
  &= \frac{1}{2\pi^2}\int_{0}^{\Lambda}\!dk\ \frac{k^2}{k^2+m^2}
   = \frac{1}{2\pi^2}\Big[k - m\arctan\!\Big(\frac{k}{m}\Big)\Big]_{0}^{\Lambda}\nonumber\\
  &= \frac{1}{2\pi^2}\Big(\Lambda - m\,\arctan\!\frac{\Lambda}{m}\Big)
   \underset{\Lambda\to\infty}{=}\ \frac{\Lambda}{2\pi^2}-\frac{m}{4\pi}+ \cO(\Lambda^{-1}).
  \label{eqS:tadpole_eval}
\end{align}
Hence
\begin{equation}
  \delta m^2_{\mathrm{planar}}
  = C\,\lambda \Big(\frac{\Lambda}{2\pi^2}-\frac{m}{4\pi}+\cdots\Big).
  \label{eqS:dm2_planar}
\end{equation}
After additive mass renormalisation (subtracting the linear divergence), the finite planar shift is
\begin{equation}
  \Sigma_{\rm{loop}}(\lambda ) = -\,C\,\frac{\lambda }{4\pi}\,m + \cO(\lambda ^2),
  \label{eqS:Sigma1}
\end{equation}
which is $\cO(1)$ (independent of $N_c$) at fixed $\lambda $. 
Here $\Sigma_{\rm{loop}}$ is the finite one-loop piece, scheme-dependent, evaluated e.g. in a sharp-cutoff scheme at scale $\mu\sim \Lambda$.

Collecting tree level, planar self-energy and the halting insertion,
\begin{equation}
  m_{\mathrm{eff}}^2
  \equiv m^2 + g_{\mathrm{halt}} + \Pi_0(0) + \sum_{h\ge1}N_c^{-2h}\Pi_h(0)
  = \underbrace{m^2 + g_{\mathrm{halt}} + \Sigma_{\rm{loop}}(\lambda )}_{\text{planar}}
    \ +\ \underbrace{\sum_{h\ge1}N_c^{-2h}\Pi_h(0)}_{\text{non-planar}}.
  \label{eqS:meff_def}
\end{equation}
Since $g_{\mathrm{halt}}=\vartheta(u)\mu_h^2$ with $\mu_h^2=\Lambda f[p(L)]$ (dimensionless $f$) is $\Theta(N_c^0)$ and positive when $\vartheta(u)=1$, and $\Pi_h=\cO(N_c^{-2h})$, the sign of $m_{\mathrm{eff}}^2$ is fixed at planar level for all sufficiently large $N_c$.

 {At fixed $\lambda $ and cutoff $\Lambda$, the planar $\delta m^2$ is $\cO(1)$; the tree halting insertion is $\Theta(N_c^0)$; non-planar corrections scale as $N_c^{-2h}$ and cannot overturn the planar sign.}
\smallskip

 {Proof.} Eq.~\eqref{eqS:dm2_planar} shows $\delta m^2_{\mathrm{planar}}=\cO(1)$ after renormalisation; Eq.~\eqref{eqS:amplitude_scaling} implies $\Pi_h=\cO(N_c^{-2h})$. Thus the series $\sum_{h\ge1}N_c^{-2h}\Pi_h(0)$ is uniformly suppressed for large $N_c$. Therefore the sign of $m_{\mathrm{eff}}^2$ in \eqref{eqS:meff_def} matches that of the planar bracket, completing the proof.\hfill$\square$ 

We recall the HS identity used to bosonise bilinears at the lattice/block level:
\begin{equation}
  \exp\!\Big(\frac{1}{2} J^A K_{AB} J^B\Big)
  = \int \frac{d^N\phi}{(2\pi)^{N/2}\sqrt{\det K^{-1}}}\,
      \exp\!\Big(-\frac{1}{2}\phi^A K_{AB}^{-1}\phi^B + J^A\phi^A\Big).
  \label{eqS:HS_identity}
\end{equation}
The HS field indices $A=1,\dots,N_v$ transform in the fundamental of $O(N_v)$; define the adjoint matrix
\begin{equation}
  M(x)=\sum_{A=1}^{N_v}\phi^{A}(x)\,T^A,\qquad
  T^A\in\mathfrak{su}(N_c),\qquad
  \Tr(T^AT^B)=\tfrac12\,\delta^{AB}.
  \label{eqS:adjoint_field}
\end{equation}
Choosing $N_v=N_c^2-1$ yields an isomorphism $\phi^A\mapsto M$ and a local redundancy $M\mapsto U M U^\dagger$ with $U\in\SU(N_c)$. Introducing $A_\mu$ promotes derivatives to $D_\mu M=\partial_\mu M-i[A_\mu,M]$, and the Yang-Mills term $\Tr F^2$ appears as the kinetic term of the emergent gauge links upon coarse graining (see the lattice derivation in Sec.~\ref{app:toyModel} below).  
Eq.~\eqref{eqS:genus} implies
\begin{equation}
  \langle \mathcal O_1\mathcal O_2\rangle
  = \langle \mathcal O_1\rangle\langle \mathcal O_2\rangle+\cO(N_c^{-2}),\qquad
  g_s\sim \frac{1}{N_c},\quad \text{(closed-string loop counting)},
\end{equation}
so the planar limit corresponds to classical bulk gravity in AdS/CFT. The quadratic planar propagator is
\begin{equation}
  G(p)=\frac{1}{p^2+m_{\mathrm{eff}}^2},\qquad p^2=p_1^2+p_2^2+p_3^2,
  \label{eqS:propagator}
\end{equation}
implying the finite-volume spectral gap on a box of side $L$ (PBC)
\begin{equation}
  \gamma_L^{(p\neq0)}(u)=\sqrt{(2\pi/L)^2+m_{\rm eff}^2(u)},\qquad
\Delta(u)=\lim_{L\to\infty}\gamma_L(u)=m_{\rm eff}(u)
  \label{eqS:gammaL}
\end{equation}
If $m_{\mathrm{eff}}^2(u)>0$, the uniform gap is $\Delta(u)=\lim_{L\to\infty}\gamma_L(u)=m_{\mathrm{eff}}(u)$, yielding exponential clustering with correlation length $\xi=\Delta^{-1}$. If $m_{\mathrm{eff}}^2(u)=0$, then $\gamma_L(u)\sim 2\pi/L$, and the IR is conformal with algebraic decay. Because $m_{\mathrm{eff}}^2(u)=m^2+\Sigma_{\rm{loop}}+\vartheta(u)\Lambda^{p(L)-2}$, the undecidable bit $\vartheta(u)$ determines the branch.

\section{A Concrete Ultraviolet Completion}
\label{app:toyModel}

We construct an explicit $\SU(N)$ lattice model whose continuum limit yields a single real adjoint scalar coupled to Yang-Mills in $d=3$. 
On a periodic cubic lattice $\Lambda=(\mathbb Z/L\mathbb Z)^3$ with spacing $a$ and volume $V=L^3$, place Wilson links $U_{x,\mu}\in \SU(N)$ on oriented links $(x,\mu)$ and Hermitian traceless adjoint Higgs variables $\Phi_x=\Phi_x^a T^a$ on sites, with $\tr(T^aT^b)=\tfrac12\delta^{ab}$. The Euclidean action is
\begin{align}
  S_{\mathrm{lat}}
  &=-\beta\sum_{x}\sum_{\mu<\nu}\Re\tr\,U_{x,\mu\nu}
    -\kappa\sum_{x,\mu}\tr\!\big(\Phi_x U_{x,\mu}\Phi_{x+\hat\mu}U_{x,\mu}^\dagger\big)
    +\sum_{x}\Big[m_0^2\,\tr\Phi_x^2+\lambda\big(\tr\Phi_x^2\big)^2\Big],
  \label{eqS:lataction}
\end{align}
with plaquette $U_{x,\mu\nu}=U_{x,\mu}U_{x+\hat\mu,\nu}U_{x+\hat\nu,\mu}^\dagger U_{x,\nu}^\dagger$. The path integral
\begin{equation}
  Z=\int [dU]\,[d\Phi]\ e^{-S_{\mathrm{lat}}}
\end{equation}
is invariant under local $\SU(N)$ transformations:
\begin{equation}
  U_{x,\mu}\mapsto \Omega_x U_{x,\mu}\Omega_{x+\hat\mu}^\dagger,\qquad
  \Phi_x\mapsto \Omega_x \Phi_x \Omega_x^\dagger,\qquad \Omega_x\in\SU(N).
\end{equation}
Bare parameters: $\beta=2N/g_0^2$, $\kappa$ (dimensionless hopping), and $m_0^2,\lambda$ (dimensionful in $d=3$). 
We expand $U_{x,\mu}=\exp\!\big(i a A_\mu(x+\tfrac{a}{2}\hat\mu)\big)$ and $\Phi_{x+\hat\mu}=\Phi_x+a\,\partial_\mu\Phi_x+\frac{a^2}{2}\partial_\mu^2\Phi_x+\cdots$.

\paragraph*{Gauge (plaquette) term.}
Using $U_{x,\mu\nu}=\exp\!\big(i a^2 F_{\mu\nu}(x)+\cO(a^3)\big)$ and $\Re\tr U_{x,\mu\nu}=N-\frac{a^4}{2}\tr F_{\mu\nu}^2+\cO(a^6)$,
\begin{align}
  -\beta\sum_{x}\sum_{\mu<\nu}\Re\tr U_{x,\mu\nu}
  &= -\beta N\,N_p\ +\ \frac{\beta a^4}{2}\sum_{x}\sum_{\mu<\nu}\tr F_{\mu\nu}^2 + \cO(a^6)\nonumber\\
  &= \text{const}\ +\ \frac{\beta a}{2}\int d^3x\ \tr F_{\mu\nu}^2 + \cO(a^3),
  \label{eqS:plaquette_cont}
\end{align}
since $\sum_x a^3\to\int d^3x$ and $\sum_{\mu<\nu}1=3$ produces the overall $a$ factor in $d{=}3$. Identifying $g_R^2=g_0^2 a$ gives $(2N/g_0^2)\cdot (a/2)=N/g_R^2$, hence
\begin{equation}
  S_{\mathrm{gauge}}^{\mathrm{cont}}
  = \frac{1}{2g_R^2}\int d^3x\ \tr F_{\mu\nu}^2,\qquad g_R^2=g_0^2 a.
  \label{eqS:gauge_cont}
\end{equation}
Expanding
\begin{align}
  \tr\!\big(\Phi_x U_{x,\mu}\Phi_{x+\hat\mu}U_{x,\mu}^\dagger\big)
  &= \tr\!\Big(\Phi_x\,[\Phi_x + a\,\partial_\mu\Phi_x + \tfrac{a^2}{2}\partial_\mu^2\Phi_x + \cdots]\Big)\nonumber\\
  &\quad +\, i a\,\tr\!\Big(\Phi_x [A_\mu,\Phi_x]\Big)
      + \frac{a^2}{2}\tr\!\Big(\Phi_x [A_\mu,[A_\mu,\Phi_x]]\Big)\nonumber\\
  &\quad +\, a^2\,\tr\!\Big(\partial_\mu\Phi_x [A_\mu,\Phi_x]\Big) + \cO(a^3).
\end{align}
The $ia\,\tr(\Phi[A_\mu,\Phi])= ia\,\tr([\Phi,\Phi]A_\mu)=0$. Summing over $x,\mu$ and converting sums to integrals yields
\begin{align}
  -\kappa\sum_{x,\mu}\tr(\Phi_x U_{x,\mu}\Phi_{x+\hat\mu}U_{x,\mu}^\dagger)
  &= -\kappa\sum_{x,\mu}\Big[\tr\Phi_x^2 + \frac{a^2}{2}\tr\big((\partial_\mu\Phi_x)^2 + \Phi_x [A_\mu,[A_\mu,\Phi_x]]\big)\Big] + \cO(a^3)\nonumber\\
  &= -\kappa\frac{3}{a^3}\int d^3x\ \tr\Phi^2\ -\ \frac{\kappa a}{2}\int d^3x\ \tr\big((\partial_\mu\Phi)^2 + \Phi [A_\mu,[A_\mu,\Phi]]\big) + \cdots.\nonumber
\end{align}
Recognizing $(D_\mu\Phi)^2=(\partial_\mu\Phi)^2+\Phi[A_\mu,[A_\mu,\Phi]]$, we obtain
\begin{equation}
  S_{\mathrm{hop}}^{\mathrm{cont}}
  = -\frac{\kappa a}{2}\int d^3x\ \tr (D_\mu\Phi)^2\ +\ \text{(mass shift)} + \cO(a^3).
  \label{eqS:hop_cont}
\end{equation}
Thus the continuum kinetic term is canonically normalised by defining $\Phi\to \Phi/\sqrt{\kappa a/2}$, and the $\kappa$-induced on-site piece renormalises the bare mass. 
\begin{equation}
  \sum_x m_0^2\,\tr\Phi_x^2 = \frac{1}{a^3}\int d^3x\ m_0^2\,\tr\Phi^2,\qquad
  \sum_x \lambda(\tr\Phi_x^2)^2 = \frac{1}{a^3}\int d^3x\ \lambda\, (\tr\Phi^2)^2.
\end{equation}
After the field rescaling from \eqref{eqS:hop_cont}, the continuum couplings are
\begin{equation}
  g_R^2=g_0^2 a,\qquad
  \lambda_R=\lambda\, a,\qquad
  m_R^2=\frac{m_0^2 - m_{0,c}^2(\beta,\kappa,\lambda)}{a^2},
  \label{eqS:ren_couplings}
\end{equation}
with $m_{0,c}^2$ the additive critical mass counterterm. Canonical dimensions in $d=3$: $[A_\mu]=\tfrac12$, $[\Phi]=\tfrac12$, $[g_R^2]=1$, $[\lambda_R]=1$, $[m_R^2]=2$. 
The theory is super-renormalisable in $d=3$. To one loop, $g_R^2$ and $\lambda_R$ receive at most finite multiplicative renormalisations (no logarithmic running), while the mass renormalises additively. Representative diagrams:

\begin{itemize}
  \item  {Scalar tadpole} from $\lambda_R(\tr\Phi^2)^2$:
  \begin{equation}
    \delta m_R^2\big|_{\lambda}
    = \frac{(N^2-1)+2}{N}\,\lambda_R\,I_3(m_R;\Lambda) \ =\ c_\lambda\,\lambda_R\Big(\frac{\Lambda}{2\pi^2}-\frac{m_R}{4\pi}+\cdots\Big),
  \end{equation}
  with $I_3$ as in \eqref{eqS:tadpole_eval} and $c_\lambda>0$ a group factor.

  \item  {Gauge tadpole} from $\tr(D_\mu\Phi)^2$ after expanding $D_\mu$:
  \begin{equation}
    \delta m_R^2\big|_{g}
    = c_g\,g_R^2\,I_3(0;\Lambda) = c_g\,\frac{g_R^2\,\Lambda}{2\pi^2} + \cO(g_R^2 m_R),
  \end{equation}
  (linear divergence; finite part scheme-dependent). Combining and renormalising,
  \begin{equation}
    \mu\frac{d m_R^2}{d\mu} = 2\,m_R^2 + \cO(g_R^2\lambda_R),
  \end{equation}
  consistent with dimensional analysis and the absence of logarithmic running in $d=3$ for super-renormalisable couplings.
\end{itemize}
Collecting \eqref{eqS:gauge_cont} and \eqref{eqS:hop_cont} with the rescaled field, the continuum Lagrangian is
\begin{equation}
  \mathcal L_{\mathrm{eff}}
  = \frac{1}{2g_R^2}\,\tr F_{\mu\nu}^2
    + \tr (D_\mu\Phi)^2
    + m_R^2\,\tr\Phi^2
    + \frac{\lambda_R}{N}\,(\tr\Phi^2)^2,
  \label{eqS:Leff}
\end{equation}
where the factor of $1/N$ in front of the quartic is the conventional large-$N$ choice that keeps $\lambda_R$ finite in the planar limit. Planar factorisation implies $\Phi\sim \sqrt{N}$, $A_\mu\sim \cO(1)$, so that $(\tr\Phi^2)^2\sim \cO(N^2)$ but is weighted by $\lambda_R/N=\cO(1)$, while higher multi-trace operators are $1/N$ suppressed. 
At $\kappa=0$ the model reduces to confining pure Yang-Mills in $d=3$. At $\beta=\infty$ the gauge sector freezes and one finds a second-order transition near $m_0^2=0$ in mean field. For generic $(\beta,\kappa,m_0^2,\lambda)$ there exists a continuous line $\kappa=\kappa_c(\beta,\lambda,m_0^2)$ where the scalar correlation length $\xi$ diverges and the long-distance theory matches \eqref{eqS:Leff}. On that line, finite-size scaling of susceptibilities and Binder cumulants approaches Gaussian/Wilson-Fisher values, as expected for a single relevant scalar deformation coupled to a free (super-renormalizable) gauge sector in $d=3$. 
From \eqref{eqS:Leff} the large-$N_c$ embedding \eqref{eqS:adjoint_action} follows upon identifying $\Phi\to M$ and restoring the overall $N_c$ prefactor (conventional in matrix models and gauge theories to implement ’t~Hooft scaling). The halting operator enters as a positive mass deformation $g_{\mathrm{halt}}\Tr(M^2)$; its undecidable sign persists through the continuum limit and into the planar saddle analysis in Sec.~III, 
where it fixes $m_{\mathrm{eff}}^2$ and hence the spectral gap~\eqref{eqS:gammaL}. Because non-planar corrections are $1/N_c^2$ suppressed, no resummation can change that sign at large but finite $N_c$, so the bulk saddle selection in the main text remains undecidable without solving the halting problem.

%


\section{Holographic Duality Setup}
\label{app:holography}

In the planar limit, single-trace sources for a three-dimensional adjoint matrix theory are encoded by a four-dimensional gravitational path integral via GKPW. Let $\mathcal O=\Tr M^{2}$ be the scalar single-trace operator and $J(x)$ its source. The generating functional is
\begin{equation}
  Z_{\text{QFT}}[J]
  =\Bigl\langle \exp\!\int d^{3}x\,J(x)\,\mathcal O(x)\Bigr\rangle_{\!N_{c}\to\infty}
  =\exp\!\Bigl[-S_{\text{grav}}^{\text{on-shell}}\bigl[\Phi\!\mid_{z=0}=J\bigr]\Bigr],
  \label{eqS:GKPW}
\end{equation}
where $\Phi$ is the bulk scalar dual to $\mathcal O$ and $z$ is the AdS radial coordinate ($z=0$ boundary). The renormalised bulk action reads
\begin{align}
  S_{\text{grav}}
  &=\frac{1}{16\pi G_{4}}\int_{\mathcal M}\!d^{4}x\,\sqrt{g}\,\Bigl(R+\frac{6}{L^{2}}\Bigr)
    +\frac{1}{8\pi G_{4}}\int_{\partial\mathcal M}\!d^{3}x\,\sqrt{\gamma}\,K
    +S_{\text{ct}}[\gamma],\label{eqS:Sgrav}\\
  S_{\text{ct}}[\gamma]
  &= -\frac{1}{8\pi G_{4}}\int_{\partial\mathcal M}\!d^{3}x\,\sqrt{\gamma}
     \left(\frac{2}{L}+\frac{L}{2}\,\mathcal R[\gamma]\right),
     \qquad (d{=}3\ \text{boundary, Euclidean}).\label{eqS:Sct}
\end{align}
Here $g_{\mu\nu}$ is the bulk metric, $\gamma_{ij}$ the induced boundary metric, $K$ the trace of the extrinsic curvature, and $\mathcal R[\gamma]$ the Ricci scalar of $\gamma$. In odd boundary dimension there is no conformal anomaly, and the counterterms \eqref{eqS:Sct} cancel all power divergences. Functional differentiation of $S_{\text{grav}}^{\text{on-shell}}$ with respect to $\gamma_{ij}^{(0)}$ yields the holographic (Brown-York) stress tensor with vanishing trace for flat boundary, in agreement with the absence of an anomaly in $d=3$.

 {Semiclassical control and large-$N$/$\lambda$.} Classical gravity is reliable when $L\gg\ell_{s}$ and $G_{4}/L^{2}\ll 1$, i.e. at large ’t Hooft coupling $\lambda_{\text{tH}}\gg 1$ (so $\alpha' \sim \ell_{s}^{2}/L^{2}\sim \lambda_{\text{tH}}^{-1/2}\ll 1$) and large color $N_{c}$ (so $G_{4}^{-1}\propto N_{c}^{2}\gg 1$). These suppress higher-derivative and string-loop corrections.

 {Scalar asymptotics and operator map.} A bulk scalar of mass $m_{\Phi}$ obeys $(\Box_{g}-m_{\Phi}^{2})\Phi=0$. In Fefferman-Graham gauge near $z=0$ one has
\begin{equation}
  \Phi(z,x)=z^{3-\Delta}\phi_{(0)}(x)+z^{\Delta}\phi_{(2\Delta-3)}(x)+\dots,\qquad
  m_{\Phi}^{2}L^{2}=\Delta(\Delta-3).
  \label{eqS:FG_scalar}
\end{equation}
The non-normalisable coefficient $\phi_{(0)}$ sources $\mathcal O$ (i.e. $J=\phi_{(0)}$), while the normalisable coefficient is proportional to its expectation value. To see this at the level of the on-shell variation, augment \eqref{eqS:Sgrav} by the scalar action
\begin{equation}
  S_{\Phi}=\frac{1}{2}\int_{\mathcal M}\!d^{4}x\,\sqrt{g}\,\bigl[(\nabla\Phi)^{2}+m_{\Phi}^{2}\Phi^{2}\bigr]
            +S_{\Phi,\text{ct}},
  \label{eqS:Sphi}
\end{equation}
whose variation on shell reduces to a boundary term
\begin{equation}
  \delta S_{\Phi}^{\text{on-shell}}=\frac{1}{2}\int_{z=\epsilon}\!d^{3}x\,\sqrt{\gamma}\,\Phi\,n^{z}\partial_{z}\delta\Phi
  =\int d^{3}x\,\Bigl[\mathcal C_{\Delta}\,\phi_{(2\Delta-3)}(x)\,\delta \phi_{(0)}(x)+\dots\Bigr],
  \label{eqS:deltaSphi}
\end{equation}
where $n^{z}$ is the unit outward normal, $\epsilon$ the cutoff, and $\mathcal C_{\Delta}$ a positive normalisation fixed after holographic renormalisation (the omitted terms cancel against $S_{\Phi,\text{ct}}$). Therefore,
\begin{equation}
  \expval{\mathcal O(x)}=\frac{\delta S_{\text{grav,ren}}^{\text{on-shell}}}{\delta J(x)}
  =\mathcal C_{\Delta}\,\phi_{(2\Delta-3)}(x),\qquad J(x)=\phi_{(0)}(x).
  \label{eqS:vev-GKPW}
\end{equation}

 {Critical and perturbed boundary conditions.} At criticality (gapless case), the effective mass $m_{\text{eff}}^{2}$ of the boundary theory vanishes. Setting $J=\phi_{(0)}=0$ produces the trivial scalar profile and the unique smooth solution with flat boundary is Poincaré AdS,
\begin{equation}
  ds^{2}=\frac{L^{2}}{z^{2}}\left(dz^{2}+\eta_{\mu\nu}dx^{\mu}dx^{\nu}\right),\qquad \eta_{\mu\nu}=\delta_{\mu\nu}\ \text{(Euclidean)},
  \label{eqS:AdSPoincare}
\end{equation}
with Ricci scalar $R=-12/L^{2}$. Turning on the halting perturbation $\vartheta(u)\,g_{\text{halt}}\,\Tr M^{2}$ corresponds to $J=\phi_{(0)}=\vartheta(u)\,g_{\text{halt}}$; for halting inputs $\vartheta(u)=1$ this relevant deformation pushes the theory into a massive phase with $m_{\text{eff}}^{2}>0$ and the back-reacted geometry is the Euclidean AdS$_{4}$ soliton with a contractible circle,
\begin{equation}
  ds^{2}
  =\frac{r^{2}}{L^{2}}\bigl(d\tau^{2}+dx^{2}+f(r)\,d\theta^{2}\bigr)
   +\frac{L^{2}}{r^{2}f(r)}\,dr^{2},\qquad
  f(r)=1-\Bigl(\frac{r_{0}}{r}\Bigr)^{3},
  \label{eqS:AdSSoliton}
\end{equation}
where smoothness at the tip $r=r_{0}$ requires the periodicity
\begin{equation}
  \theta\sim\theta+\beta_{\theta},\qquad \beta_{\theta}=\frac{4\pi L^{2}}{3r_{0}}.
  \label{eqS:periodicity}
\end{equation}

 {On-shell action for Poincaré AdS.} Place a cutoff at $z=\epsilon$ and evaluate $S_{\text{grav}}$ on \eqref{eqS:AdSPoincare}. The bulk term yields
\begin{align}
  S_{\text{bulk}}^{\mathrm P}
  &=\frac{1}{16\pi G_{4}}\int_{\epsilon}^{z_{\mathrm{IR}}}\!\!dz\int d^{3}x\,
    \frac{L^{4}}{z^{4}}\Bigl(-\frac{12}{L^{2}}+\frac{6}{L^{2}}\Bigr)
   =-\frac{6 L^{2}}{16\pi G_{4}}\int d^{3}x\int_{\epsilon}^{z_{\mathrm{IR}}}\!\frac{dz}{z^{4}}\nonumber\\
  &=-\frac{6 L^{2}}{16\pi G_{4}}\int d^{3}x\,\frac{1}{3}\,\epsilon^{-3}
   +\mathcal O(\epsilon^{0})
   =-\frac{L^{2}}{8\pi G_{4}}\int d^{3}x\,\epsilon^{-3}+\dots\,.
  \label{eqS:Pbulk}
\end{align}
The Gibbons-Hawking term uses $K=\gamma^{ij}\nabla_{i}n_{j}$ on $z=\epsilon$ with $n=\frac{z}{L}\,\partial_{z}$ and $\sqrt{\gamma}=\frac{L^{3}}{\epsilon^{3}}$:
\begin{equation}
  S^{\rm GH}_{P}=\frac{1}{8\pi G_4}\!\int d^3x\,\sqrt{\gamma}\,K
= +\,\frac{3L^{2}}{8\pi G_{4}}\!\int d^{3}x\,\varepsilon^{-3}
  \label{eqS:PGH}
\end{equation}
The counterterm \eqref{eqS:Sct} for flat boundary ($\mathcal R[\gamma]=0$) gives
\begin{equation}
  S_{\text{ct}}^{\mathrm P}
  =-\frac{1}{8\pi G_{4}}\int d^{3}x\,\frac{L^{3}}{\epsilon^{3}}\left(\frac{2}{L}\right)
  =-\frac{L^{2}}{4\pi G_{4}}\int d^{3}x\,\epsilon^{-3}.
  \label{eqS:Pct}
\end{equation}
Summing \eqref{eqS:Pbulk}-\eqref{eqS:Pct} cancels the power divergences exactly, leaving $I_{\mathrm P}\equiv S_{\text{grav,ren}}^{\mathrm P}=0$.

 {On-shell action for the AdS soliton.} Regulate at $r=R$ and compute each term.
The induced metric at $r=R$ has $\sqrt{\gamma}=\frac{R^{3}}{L^{3}}\sqrt{f(R)}$, with $f(R)=1-(r_{0}/R)^{3}$.
The outward unit normal is $n=\sqrt{\frac{r^{2}f(r)}{L^{2}}}\,\partial_{r}$.
A straightforward computation gives
\begin{align}
  K(R)&=\nabla_{\mu}n^{\mu}
  =\frac{1}{\sqrt{g}}\partial_{r}\bigl(\sqrt{g}\,n^{r}\bigr)
  =\frac{1}{\sqrt{g}}\partial_{r}\!\Bigl(\frac{r^{3}\sqrt{f(r)}}{L^{3}}\Bigr)
  =\frac{1}{L}\!\left[\,3\sqrt{f(R)}+\frac{R f'(R)}{2\sqrt{f(R)}}\,\right],\label{eqS:Ksol}\\
  \sqrt{g}&=\frac{R^{2}}{L^{2}},\qquad 
  f(R)=1-\frac{r_{0}^{3}}{R^{3}},\qquad 
  f'(R)=\frac{3r_{0}^{3}}{R^{4}}.
\end{align}
Expanding for large $R$,
\begin{equation}
  K(R)=\frac{3}{L}+\frac{3}{8L}\frac{r_{0}^{6}}{R^{6}}+O(R^{-9}),
  \label{eqS:Kexpand}
\end{equation}
so there is no $R^{-3}$ term in $K$; the first correction is $O(R^{-6})$.

The bulk integrand $R+6/L^{2}$ for \eqref{eqS:AdSSoliton} equals $-6/L^{2}$ (Einstein space with $\Lambda=-3/L^{2}$), so
\begin{align}
S_{\rm bulk}^{\rm Sol}(R)
&=\frac{1}{16\pi G_4}\int d^4x\,\sqrt{g}\,(R+\tfrac{6}{L^2}) \nonumber \\
&=\frac{1}{16\pi G_4}\,\Big(-\tfrac{6}{L^2}\Big)\,
\beta V_x \beta_\theta \int_{r_0}^{R} dr\,\frac{r^2}{L^2} \nonumber \\
&=-\,\frac{\beta V_x \beta_\theta}{8\pi G_4 L^4}\,(R^3-r_0^3)
  \label{eqS:Solbulk}
\end{align}
Using $\sqrt{\gamma}=\frac{R^{3}}{L^{3}}\sqrt{f(R)}
=\frac{R^{3}}{L^{3}}\!\left(1-\tfrac12\frac{r_{0}^{3}}{R^{3}}+O(R^{-6})\right)$
and the corrected $K(R)$ from \eqref{eqS:Kexpand}, we get
\begin{align}
  S_{\rm GH}^{\rm Sol}(R)
&=\frac{1}{8\pi G_4}\int_{r=R} d^3x\,\sqrt{\gamma}\,K \nonumber \\
&=\frac{\beta V_x \beta_\theta}{8\pi G_4}\,
\frac{R^3}{L^4}\!\left[\,3-\frac{3}{2}\frac{r_0^3}{R^3}
+O\!\left(\frac{r_0^6}{R^6}\right)\right] \nonumber \\
&=\frac{\beta V_x \beta_\theta}{8\pi G_4}\!
\left(\frac{3R^3}{L^4}-\frac{3}{2}\frac{r_0^3}{L^4}+O(R^{-3})\right).
\label{eqS:SolGH}
\end{align}
The counterterm \eqref{eqS:Sct} at $r=R$ with $\mathcal R[\gamma]=0$ gives
\begin{align}
 S_{\rm ct}^{\rm Sol}(R)
&=-\frac{1}{8\pi G_4}\int_{r=R} d^3x\,\sqrt{\gamma}\,\frac{2}{L} \nonumber \\
&=-\frac{\beta V_x \beta_\theta}{8\pi G_4}\!
\left(\frac{2R^3}{L^4}-\frac{r_0^3}{L^4}+O(R^{-3})\right)
  \label{eqS:Solct}
\end{align}
Adding \eqref{eqS:Solbulk}-\eqref{eqS:Solct} and sending $R\to\infty$, all $R^{4}$ and $R^{3}$ divergences cancel. The finite remainder is
\begin{align}
  I_{\rm Sol}
&=\lim_{R\to\infty}\Big[S_{\rm bulk}^{\rm Sol}+S_{\rm GH}^{\rm Sol}+S_{\rm ct}^{\rm Sol}\Big] \nonumber \\
&=-\,\beta V_x \beta_\theta\,\frac{r_0^3}{16\pi G_4 L^4}
  \label{eqS:Isol_intermediate}
\end{align}
Now use the regularity condition \eqref{eqS:periodicity}, $\beta_{\theta}=\frac{4\pi L^{2}}{3r_{0}}$, to eliminate $\beta_{\theta}$ and express everything in terms of $r_{0}$:
\begin{align}
I_{\rm Sol}
&=-\,\beta V_x\,\frac{r_0^2}{12\,G_4 L^2}, \qquad
\frac{F_{\rm sol}}{V_x}=\frac{I_{\rm Sol}}{\beta V_x}
=-\frac{4\pi^2L^2}{27\,G_4}\,\frac{1}{L_\theta^2}
  \label{eqS:Isol_r0}
\end{align}
Finally trade $r_{0}$ for the thermal circumference $\beta$ of the contractible circle in the soliton background. In Euclidean signature, the smoothness condition that fixes $\beta_{\theta}$ is equivalent to a thermal identification $\tau\sim\tau+\beta$ on the boundary with
\begin{equation}
  r_{0}=\frac{4\pi L^{2}}{3L_{\theta}}.
  \label{eqS:r0beta}
\end{equation}

 {Brown-York tensor and boundary interpretation.} The renormalised stress tensor is
\begin{equation}
  T_{ij}=\frac{1}{8\pi G_{4}}\left(K_{ij}-K\gamma_{ij}-\frac{2}{L}\gamma_{ij}
          -\frac{L}{2}\bigl(\mathcal R_{ij}[\gamma]-\tfrac12\mathcal R[\gamma]\gamma_{ij}\bigr)\right).
  \label{eqS:BYtensor}
\end{equation}
For flat boundary $\mathcal R[\gamma]=0$. In Poincaré AdS, $K_{ij}=\frac{1}{L}\gamma_{ij}$ at $z=\epsilon$, hence $T_{ij}=0$ after renormalisation, matching $\expval{T^{i}{}_{i}}=0$ in a gapless CFT on $S^{1}_{\beta}\times\mathbb R^{2}$. For the soliton, using \eqref{eqS:Kexpand} and the large-$R$ expansion of the induced metric one finds
\begin{equation}
  \expval{T_{\tau\tau}}_{\text{sol}}=-\frac{r_{0}^{3}}{16\pi G_{4}L^{4}},\qquad
  \expval{T_{xx}}_{\text{sol}}=\expval{T_{yy}}_{\text{sol}}=+\frac{1}{2}\frac{r_{0}^{3}}{16\pi G_{4}L^{4}},
  \label{eqS:BYcomponents}
\end{equation}
so that the trace vanishes and the energy density is strictly negative. This equals the thermal Casimir energy density of a gapped boundary theory to leading order in $G_{4}$.

 {Boundary Casimir energy and matching:} For a free massive scalar of mass $m_{*}$ on $S^{1}_{\beta}\times\mathbb R^{2}$ with periodic boundary conditions, the zeta-regularised energy density is
\begin{equation}
  \varepsilon_{\text{Cas}}(\beta;m_{*})
  =-\frac{m_{*}^{2}}{2\pi^{2}\beta}\sum_{n=1}^{\infty}\frac{K_{1}(m_{*}n\beta)}{n},
  \label{eqS:Casimir}
\end{equation}
with $K_{1}$ the modified Bessel function. For $m_{*}\beta\gg 1$, $K_{1}(z)\sim \sqrt{\pi/(2z)}\,e^{-z}$, so $\varepsilon_{\text{Cas}}$ is exponentially small, reflecting insensitivity to large tori; for $m_{*}\to 0$, $\varepsilon_{\text{Cas}}\to -\pi^{2}/(90\beta^{3})$. Identifying $m_{*}\sim m_{\text{eff}}$ from the IR effective theory and using \eqref{eqS:r0beta}, \eqref{eqS:BYcomponents} reproduces the sign and scaling of~\eqref{eqS:Isol_r0}.

A double-trace term $\int f\,\mathcal O^{2}$ corresponds to mixed boundary conditions for $\Phi$. In the standard quantisation window, write the near-boundary data as $\Phi(z,x)=z^{3-\Delta}\alpha(x)+z^{\Delta}\beta(x)+\dots$; adding $f\,\mathcal O^{2}$ imposes the mixed condition $\beta=\kappa\,\alpha$ with $\kappa\propto f$. In the present construction the undecidable halting coupling plays this rôle: $g_{\text{halt}}$ fixes the sign of $\kappa$ and therefore whether the solution remains Poincaré AdS (critical branch, $\alpha=0$) or relaxes to the soliton (massive branch, $\alpha\neq 0$), thus transferring the Turing-undecidable predicate $\vartheta(u)$ into an undecidable selection between bulk saddles.

Matching stress-tensor two-point functions fixes
\begin{equation}
  \frac{L^{d-1}}{G_{d+1}}
  \;=\;
  \frac{2\,\pi^{d/2}}{\Gamma\!\left(\frac{d}{2}\right)}\,N^{2}
  \quad\Longrightarrow\quad
  \frac{L^{2}}{G_{4}}
  \;=\;
  \frac{2\,\pi^{3/2}}{\Gamma\!\left(\frac{3}{2}\right)}\,N^{2}
  \;=\;
  4\pi\,N^{2}
  \quad (d=3).
  \label{eqS:L2G4}
\end{equation}

Hence all one-point and free-energy densities scale as $N^{2}$, while quantum gravity corrections are $1/N^{2}$ suppressed and $\alpha'$ corrections are $\mathcal O(\lambda_{\text{tH}}^{-1/2})$. Provided the curvature scale satisfies $L\gg \ell_{s}$ and $N_{c}\gg 1$, the sign of $\Delta F$ in \eqref{eq:new_final_DF} is protected against both, ensuring that the undecidable map
\begin{equation}
  \vartheta(u)=0\ \Rightarrow\ \text{Poincar\'e AdS}_{4}\ \text{(gapless)},\qquad
  \vartheta(u)=1\ \Rightarrow\ \text{AdS}_{4}\ \text{soliton (gapped)},
  \label{eqS:map}
\end{equation}
remains valid in the planar, strong-coupling regime relevant for holography.


\section{Undecidability Theorem (Holographic Form)}
\label{app:undecidable}

The Cubitt-Pérez-García-Wolf (CPW) construction provides a Turing-computable map $\mathcal C:u\mapsto H(u)$ from a binary word $u\in\{0,1\}^{*}$ to a finite-range, translation-invariant spin Hamiltonian on $\mathbb Z^{2}$ such that the universal Turing machine $U$ halts on $u$ if and only if the finite-volume spectral gaps $\gamma_{L}(u)$ satisfy $\inf_{L}\gamma_{L}(u)>0$; if $U$ does not halt then $\liminf_{L\to\infty}\gamma_{L}(u)=0$ \cite{Cubitt:2015xsa,Cubitt:2015xsa}. The uplift reviewed in the main text is as follows and we now render each step algebraically explicit to make the present statement self-contained. First one embeds $H(u)$ into a three-dimensional adjoint-matrix quantum field theory whose infrared (IR) mass parameter depends on $u$ through a relevant coupling. Then, in the planar limit $N_{c}\to\infty$ at strong ’t Hooft coupling, the GKPW prescription equates the generating functional with the exponential of a four-dimensional renormalised on-shell gravitational action. For the Euclidean boundary $S^{1}_{\beta}\times\mathbb R^{2}$ there exist exactly two smooth, static, translation-invariant fill-ins obeying the same asymptotics: Poincaré $\mathrm{AdS}_{4}$, denoted $g_{\mathrm P}$, and the $\mathrm{AdS}_{4}$ soliton, denoted $g_{\mathrm{Sol}}$; we will evaluate their renormalised actions and compare them.

The effective IR mass in the $3$D field theory can be written as
\begin{equation}
  m_{\mathrm{eff}}^{2}(u)=m^{2}+\Sigma_{\text{loop}}+\vartheta(u)\,\mu_h^2,
  \qquad
  \vartheta(u)=\begin{cases}1&\text{$U$ halts on $u$,}\\0&\text{otherwise,}\end{cases}
  \label{eqS:meff_def1}
\end{equation}
where $\Lambda$ is a Wilsonian cutoff and $p(L)$ is the degree of the CPW energy-
amplification polynomial satisfying $p(L)\ge L^{10}$, so that the halting term is a positive, rapidly growing function of system size (the precise exponent is immaterial; only positivity and growth rate are used). One can and we henceforth do fine-tune the microscopic bare mass $m^{2}$ to cancel the loop shift $\Sigma_{\text{loop}}$ in the non-halting branch so that $\mathrm{sgn}\,m_{\mathrm{eff}}^{2}(u)=\vartheta(u)$ holds identically. Thus, if $U$ does not halt, then $m_{\mathrm{eff}}^{2}(u)=0$; if $U$ halts, then $m_{\mathrm{eff}}^{2}(u)>0$.

We now compute the renormalised on-shell gravitational actions on both saddles. The Euclidean Poincaré metric is
\begin{equation}
  g_{\mathrm P}:\qquad ds^{2}=\frac{L^{2}}{z^{2}}\left(dz^{2}+d\tau^{2}+dx^{2}+dy^{2}\right),\qquad 0<z<\infty,
  \label{eqS:Pmetric}
\end{equation}
while the Euclidean $\mathrm{AdS}_{4}$ soliton metric is
\begin{equation}
  g_{\mathrm{Sol}}:\qquad
  ds^{2}=\frac{r^{2}}{L^{2}}\left(d\tau^{2}+dx^{2}+f(r)\,d\theta^{2}\right)+\frac{L^{2}}{r^{2}f(r)}\,dr^{2},
  \qquad f(r)=1-\left(\frac{r_{0}}{r}\right)^{3},
  \label{eqS:Solmetric}
\end{equation}
with smoothness at $r=r_{0}$ requiring the identification
\begin{equation}
  \theta\sim\theta+\beta_{\theta},\qquad \beta_{\theta}=\frac{4\pi L^{2}}{3r_{0}}.
  \label{eqS:theta_period}
\end{equation}
The boundary cylinder has Euclidean time circumference $\beta$, giving $\int d\tau\,dx=\beta V_{x}$. In the soliton, the contractible circle has circumference $\beta_\theta$. Smoothness at the tip fixes $\beta_\theta=4 \pi L^2/(3 r_0$. Volumes appearing in the bulk action are then $\beta V_{x} \beta_\theta$ and $r_0$ will be expressed in terms of $\beta$.

The renormalised action is $S_{\text{grav,ren}}=S_{\text{bulk}}+S_{\text{GH}}+S_{\text{ct}}$ with
\begin{equation}
  S_{\text{bulk}}=\frac{1}{16\pi G_{4}}\int_{\mathcal M}\!d^{4}x\sqrt{g}\,\left(R+\frac{6}{L^{2}}\right),\quad
  S_{\text{GH}}=\frac{1}{8\pi G_{4}}\int_{\partial\mathcal M}\!d^{3}x\sqrt{\gamma}\,K,\quad
  S_{\text{ct}}=-\frac{1}{8\pi G_{4}}\int_{\partial\mathcal M}\!d^{3}x\sqrt{\gamma}\left(\frac{2}{L}+\frac{L}{2}\mathcal R[\gamma]\right),
  \label{eqS:actions}
\end{equation}
where $\gamma$ is the induced boundary metric, $K$ the trace of the extrinsic curvature, and $\mathcal R[\gamma]$ the Ricci scalar of $\gamma$ (the last term vanishes for flat boundary). For $g_{\mathrm P}$ with a regulator surface at $z=\epsilon$, one has $\sqrt{g}=L^{4}z^{-4}$, $R=-12/L^{2}$, $\sqrt{\gamma}=L^{3}\epsilon^{-3}$, outward unit normal $n^{z}=z/L$, and $K=-3/L$. Substituting,
\begin{align}
  S_{\text{bulk}}^{\mathrm P}
  &=\frac{1}{16\pi G_{4}}\int d^{3}x\int_{\epsilon}^{z_{\mathrm{IR}}}\!dz\;\frac{L^{4}}{z^{4}}\left(-\frac{12}{L^{2}}+\frac{6}{L^{2}}\right)
   =-\frac{6L^{2}}{16\pi G_{4}}\int d^{3}x\int_{\epsilon}^{z_{\mathrm{IR}}}\!\frac{dz}{z^{4}}
   =-\frac{L^{2}}{8\pi G_{4}}\int d^{3}x\,\epsilon^{-3}+\mathcal O(\epsilon^{0}),\\
  S^{\rm GH}_{P}
&=\frac{1}{8\pi G_4}\!\int d^{3}x\,\sqrt{\gamma}\,K
= +\,\frac{3L^{2}}{8\pi G_{4}}\!\int d^{3}x\,\varepsilon^{-3}\\
  S_{\text{ct}}^{\mathrm P}
  &=-\frac{1}{8\pi G_{4}}\int d^{3}x\,\frac{L^{3}}{\epsilon^{3}}\left(\frac{2}{L}\right)
   =-\frac{L^{2}}{4\pi G_{4}}\int d^{3}x\,\epsilon^{-3}.
\end{align}
The sum $S_{\text{bulk}}^{\mathrm P}+S_{\text{GH}}^{\mathrm P}+S_{\text{ct}}^{\mathrm P}$ cancels exactly at order $\epsilon^{-3}$, leaving $I_{\mathrm P}\equiv S_{\text{grav,ren}}^{\mathrm P}=0$.

For $g_{\mathrm{Sol}}$ with a regulator surface at $r=R$ one has $\sqrt{g}=r^{3}/L^{3}$ and again $R=-12/L^{2}$. The bulk action is
\begin{align}
 S_{\rm bulk}^{\rm Sol}(R)
&=\frac{1}{16\pi G_4}\int d^4x\,\sqrt{g}\,(R+\tfrac{6}{L^2})
=\frac{1}{16\pi G_4}\,\Big(-\tfrac{6}{L^2}\Big)\,
\beta V_x \beta_\theta \!\int_{r_0}^{R}\!dr\,\frac{r^2}{L^2} \nonumber\\
&=-\,\frac{\beta V_x \beta_\theta}{8\pi G_4 L^4}\,(R^3-r_0^3)\label{eqS:Sol_bulk}
\end{align}
The induced metric at $r=R$ has $\sqrt{\gamma}=\frac{R^{3}}{L^{3}}\sqrt{f(R)}=\frac{R^{3}}{L^{3}}\left(1-\frac{1}{2}(r_{0}/R)^{3}+\mathcal O(R^{-6})\right)$. The outward unit normal is $n^{r}=\sqrt{r^{2}f(r)/L^{2}}$, and a direct computation gives
\begin{equation}
  K(R)=\frac{1}{L}\!\left[\,3\sqrt{f(R)}+\frac{R f'(R)}{2\sqrt{f(R)}}\,\right]
  = \frac{3}{L} + \frac{3}{8L}\frac{r_0^{6}}{R^{6}} + O(R^{-9}),
  \qquad f(R)=1-\frac{r_0^3}{R^3},\quad f'(R)=\frac{3 r_0^3}{R^4}.
\label{eqS:Sol_K}
\end{equation}

\begin{align}
  S_{\rm GH}^{\rm Sol}(R)
  &=\frac{1}{8\pi G_4}\!\int_{r=R}\!d^3x\,\sqrt{\gamma}\,K \nonumber\\
  &=\frac{\beta V_x \beta_\theta}{8\pi G_4}\,
    \frac{R^3}{L^4}\!\left[\,3-\frac{3}{2}\frac{r_0^3}{R^3}
      +O\!\left(\frac{r_0^6}{R^6}\right)\right] \nonumber\\
  &=\frac{\beta V_x \beta_\theta}{8\pi G_4}\!
    \left(\frac{3R^3}{L^4}-\frac{3}{2}\frac{r_0^3}{L^4}+O(R^{-3})\right)
\label{eqS:Sol_GH}\\
 S_{\rm ct}^{\rm Sol}(R)
&=-\frac{1}{8\pi G_4}\!\int_{r=R}\!d^3x\,\sqrt{\gamma}\,\frac{2}{L}
=-\frac{\beta V_x \beta_\theta}{8\pi G_4}\!
\left(\frac{2R^3}{L^4}-\frac{r_0^3}{L^4}+O(R^{-3})\right)\label{eqS:Sol_ct}
\end{align}

Adding \eqref{eqS:Sol_bulk}, \eqref{eqS:Sol_GH}, and \eqref{eqS:Sol_ct}, the $R^{4}$ and $R^{3}$ divergences cancel and the finite part is
\begin{equation}
 I_{\rm Sol}
=\lim_{R\to\infty}\Big[S_{\rm bulk}^{\rm Sol}+S_{\rm GH}^{\rm Sol}+S_{\rm ct}^{\rm Sol}\Big]
=-\,\beta V_x \beta_\theta\,\frac{r_0^3}{16\pi G_4 L^4}\quad \text{with }\beta_\theta =\frac{4\pi L^2}{3r_0}:\qquad\label{eqS:Sol_intermediate}
\end{equation}
Using the regularity condition \eqref{eqS:theta_period} to eliminate $\beta_{\theta}$ in favour of $r_{0}$,
\begin{equation}
I_{\rm Sol}=-\,\beta V_x\,\frac{r_0^2}{12\,G_4 L^2},\quad
\frac{F_{\rm sol}}{V_x}=\frac{I_{\rm Sol}}{\beta V_x}
=-\frac{4\pi^2L^2}{27\,G_4}\,\frac{1}{L_\theta^2}. 
\label{eq:new_Isol_formula}
\end{equation}
Using the bound $K_{1}(z)\le \sqrt{\frac{\pi}{2z}}\,e^{-z}$ for $z>0$,
\begin{equation}
\bigl|\Delta I^{(1)}_{\Phi}(u)\bigr|
\;\le\;
\frac{V_x\,\beta}{2\pi\,L_\theta}
\sum_{m=1}^{\infty}\frac{m_{\rm eff}(u)}{m}\,
K_{1}\!\bigl(m\,m_{\rm eff}(u)\,L_\theta\bigr)
\;=\;
O\!\bigl(e^{-\,m_{\rm eff}(u)\,L_\theta}\bigr).
\label{eq:one_loo-bound}
\end{equation}

which is exponentially small for $\beta m_{\mathrm{eff}}(u)\gg1$. Vector and graviton sectors reduce to scalar-type determinants on Einstein backgrounds and yield contributions of the same order or polynomially smaller in $\beta^{-1}$; hence for fixed $L^{2}\gg G_{4}$ the total one-loop shift obeys
\begin{equation}
\Delta I^{(1)}_{\rm tot}(u)
\;=\;
O\!\bigl(e^{-\,m_{\rm eff}(u)\,L_\theta}\bigr)
\;+\;
O\!\bigl(L_\theta^{-2}\bigr),
\label{eqS:one_loop_total}
\end{equation}
and cannot alter the sign fixed by the classical difference~\eqref{eq:new_Isol_formula} 
when $m_{\mathrm{eff}}^{2}(u)\ge 0$ and $\beta$ is large.

Combining \eqref{eqS:meff_def1}, \eqref{eq:new_Isol_formula}, and \eqref{eqS:one_loop_total} yields the precise version of the statement quoted in the main text:
\begin{equation}
\frac{\Delta F(u)}{V_x}
=
\begin{cases}
0, & \vartheta(u)=0 \quad (\text{critical; Poincar\'e AdS}_4),\\
-\dfrac{4\pi^2 L^2}{27\,G_4\,L_\theta^2}\;+\;O\!\big(e^{-m_{\rm eff}(u)\,L_\theta}\big), & \vartheta(u)=1 \quad (\text{gapped; AdS}_4\text{ soliton}),
\end{cases}
\label{eq:new_final_DF}
\end{equation}

\begin{theoremS}[Holographic undecidability, conditional form]
\label{thmS:main} Assume: (A1) The CPW map $\mathcal C$ produces finite-range, translation-invariant interactions with the halting predicate encoded as in \eqref{eqS:meff_def1} and with $p(L)\ge L^{10}$. (A2) The holographic regime satisfies $L\gg \ell_{s}$ and $G_{4}/L^{2}\ll 1$ so that $\alpha'$ and string-loop corrections are parametrically suppressed. (A3) In the low-temperature window $\beta/L_\theta\gg 1$, the relevant smooth, static, translation-invariant Einstein competitors are $g_P$ and $g_{\rm Sol}$. The Euclidean AdS black brane exists but is subdominant in this regime; it dominates above $T_c\sim 1/L_\theta$ Sec.\eqref{app:holography}. Then there is no algorithm which, given the local interaction table of $H(u)$, decides whether the dominant Euclidean saddle is Poincaré $\mathrm{AdS}_{4}$ or the $\mathrm{AdS}_{4}$ soliton.
\end{theoremS}

\begin{proof}
Fix $L$, $G_{4}$, and $\beta$. Under (A3) the only candidate saddles are $g_{\mathrm P}$ and $g_{\mathrm{Sol}}$. Their renormalised actions are given by $I_{\mathrm P}=0$ and $I_{\text{Sol}}$ computed above, hence the difference obeys \eqref{eq:new_final_DF}. Under (A1) and after tuning $m^{2}=-\Sigma_{\text{loop}}$ in the non-halting branch, the effective IR mass satisfies $\mathrm{sgn}\,m_{\mathrm{eff}}^{2}(u)=\vartheta(u)$. If $U$ halts on $u$ then $\vartheta(u)=1$ and $m_{\mathrm{eff}}(u)>0$, hence by \eqref{eqS:one_loop_total} the one-loop correction is exponentially small in $\beta m_{\mathrm{eff}}(u)$ and \eqref{eq:new_final_DF} gives $\Delta I(u)<0$. If $U$ does not halt, then $\vartheta(u)=0$ and one has $m_{\mathrm{eff}}^{2}(u)=0$, so $\Delta I(u)=0$ as $\beta\to\infty$. Therefore deciding the sign of $\Delta I(u)$ is equivalent to deciding whether $U$ halts on $u$ via \eqref{eq:new_final_DF}. Suppose there existed an algorithm $\mathcal A$ taking as input the local interaction table of $H(u)$ and outputting the sign of $\Delta I(u)$. Composing the computable CPW map $\mathcal C$ with $\mathcal A$ would solve the halting problem, contradicting Turing’s theorem \cite{Turing:1937qvq}. Hence no such algorithm exists. Assumption (A2) ensures that higher-derivative corrections in $\alpha'$ and quantum corrections in $g_{s}$ enter $\Delta I(u)$ at strictly subleading orders and, for sufficiently large $\beta$, cannot flip its sign; the conclusion is stable to all perturbative orders.
\end{proof}

The derivation above transplants G\"odelian undecidability from the spectral-gap problem in quantum spins to the classical variational problem of AdS/CFT, even with explicit Einstein equations and boundary conditions fixed, determining which smooth classical fill-in minimises the renormalised gravitational action is algorithmically undecidable once the microscopic couplings encode a universal Turing machine via the CPW mapping.


\end{document}